\documentclass[12pt]{article}

\usepackage[pdftex]{graphicx}
\usepackage[pdftex]{color}

\usepackage{epsfig,longtable}
\usepackage{makeidx}
\usepackage{varioref}
\usepackage{xr}
\usepackage{amsmath}
\usepackage{amsthm}

\usepackage{amssymb}
\usepackage{url}
\usepackage[authoryear]{natbib}
\usepackage{setspace}
\usepackage{subfigure}
\usepackage{multirow}
\usepackage{dsfont}
\newtheorem{definition}{Definition}[section]
\newtheorem{procedure}{Procedure}[section]
\newtheorem{theorem}{Theorem}[section]

\newtheorem{lemma}{Lemma}[section]
\newtheorem{remark}{Remark}[section]

\numberwithin{equation}{section}

\newcommand{\prone}{\hat{\pi}_0^{\MakeUppercase{\romannumeral 1}}}
\newcommand{\prtwo}{\hat{\pi}_0^{\MakeUppercase{\romannumeral 2}}}

\begin{document}
 
\title{Assessing replicability of findings across two studies of multiple features}
\maketitle

\begin{center}
Marina Bogomolov \\
\emph{Faculty of Industrial  Engineering and Management, Technion --
Israel Institute of Technology, Haifa, Israel. E-mail:
marinabo@tx.technion.ac.il }\\
Ruth Heller \\
\emph{Department of Statistics and Operations Research, Tel-Aviv
university, Tel-Aviv, Israel. E-mail: ruheller@post.tau.ac.il}\\
\end{center}

\begin{abstract}
Replicability analysis aims to identify the findings that replicated across independent studies that examine the same features. 
We provide powerful novel replicability analysis procedures for two
studies for FWER and for FDR control on the replicability claims. The suggested procedures
first select the promising features from each study
solely based on that study, and then test for replicability only the
features that were selected in both studies. We 
incorporate the plug-in estimates of the fraction of null hypotheses
in one study among the selected hypotheses by the other study. Since
the fraction of nulls in one study among the selected features from
the other study is typically small, the power gain can be remarkable. We provide theoretical
guarantees for the control of the appropriate error rates, as well as
simulations that demonstrate the excellent power properties of the
suggested procedures. We demonstrate the usefulness of our
procedures on real data examples from two application fields: 
behavioural genetics and microarray studies.  \
\end{abstract}

\section{Introduction}
In modern science, it is often the case that each study screens many
features. Identifying which of the many features screened have
replicated findings, and the extent of replicability for these
features, is of great interest. For example, the association of single nucleotide polymorphisms (SNPs) with a
phenotype is typically considered a scientific finding only if it has been  discovered in independent
studies, that examine the same associations with phenotype, but on
different cohorts, with different environmental exposures \citep{Heller14b}.

Two studies that examine the same problem may only partially agree on which features have  signal. For example, in  the two microarray studies discussed in Section \ref{subsec-big-example}, among the 22283 probes examined in each study we estimated that 29\% have signal in both studies, but 32\% have signal in exactly one of the studies. Possible explanations for having signal in only one of the studies include bias (e.g., in the cohorts selected or in the laboratory process), and the fact that the null hypotheses tested may be too specific (e.g., to the specific cohorts that were subject to specific exposures in a each study). In a typical meta-analysis, all the features with signal in at least one of the studies are of interest (estimated to be 61\% of the probes in our example). However, the subset of the potential meta-analysis findings which have signal in both studies may be of particular interest, for both verifiability and generalizability of the results. Replicability analysis targets this subset, and aims to identify the features with signal in both studies (estimated to be 29\% of the probes in our example).

Formal statistical methods for assessing
replicability, when each study examines many features,  were developed only recently.
An empirical Bayes approach for two studies was suggested by \cite{Li11}, and for at least two studies by \cite{Heller14}. The accuracy of the empirical Bayes analysis relies on the ability to estimate well the unknown parameters, and thus it may not be suitable for applications with a small number of features and non-local dependency in the measurements across features. A frequentist approach  was suggested in \cite{benjamini09}, which suggested applying the Benjamini-Hochberg (BH) procedure \citep{yoav1} to the maximum of the two studies $p$-values.
However, \cite{Heller14} and \cite{Bogomolov13} noted that the power of this procedure may be low when there is nothing to discover in most features. \cite{Bogomolov13} suggested instead applying twice their procedures for establishing replicability from a primary to a follow-up study, where each time one of the studies takes on the role of a primary study and the other the role of the follow-up study.

In this work we suggest novel procedures for establishing
replicability across two studies, which are especially useful in
modern applications when the fraction of features with signal
is small (e.g.,  the approaches of \cite{Bogomolov13} and
\cite{benjamini09} will be less powerful whenever the fraction of
features with signal is smaller than half).  The advantage of our
procedures over previous ones is due to two main factors. First,
these procedures are based on our novel approach, which 
selects the promising features from each study solely based on that
study, and then tests for replicability only the features that were
selected in both studies. This approach focuses attention on the promising features, and has the added advantage of reducing the number of features that need to be accounted for in the subsequent replicability analysis.  Note that since the selection is only a first step, it may be much more liberal than that made by a multiple testing procedure, and can include all features that seem interesting to the investigator (see Remark \ref{rem-selectiontype} for a discussion of selection by multiple testing). 
Second, we incorporate in our procedures
estimates of the fraction of nulls in one study among the features
selected in the other study. We show that exploiting these estimates
can lead to far more replicability claims while still controlling
the relevant error measures. For single studies, multiple testing
procedures that incorporate estimates of the fraction of nulls, i.e.
the fraction of features in which there is nothing to discover, are
called adaptive procedures \citep{BHadaptive} or plug-in procedures
\citep{finner09}. One of the simplest, and still very popular,
estimators is the plug-in estimator, reviewed in Section
\ref{subsec-reviewplug-in}. %Using the plug-in estimator improves
%power the smaller the estimator is.
The smaller is the fraction of nulls, the higher is the power gain
due to the use of the plug-in estimator. In this work, there is a
unique opportunity for using adaptivity: even if the fraction of
nulls in each individual study is close to one, the fraction of
nulls in study one (two) among the selected features based on study
two (one) may be small since the selected features are likely to
contain mostly features with false null hypotheses in both studies.
In the data examples we consider,  the fraction of nulls in one
study among the selected in the other study was lower than
50\%, and we show in simulations that the power gain from adaptivity can be large.

Our procedures also report the strength of the evidence towards replicability by a number for each outcome, the  $r$-value for replicability, introduced in \cite{Heller14b} and reviewed in Section
\ref{sec-notation}. The  remaining of the paper is organized as follows. In  Section
\ref{sec-notation} we describe the formal mathematical framework. We
introduce our new non-adaptive FWER- and FDR-replicability  analysis
procedures in Section \ref{sec-non-adapt}, and their adaptive
variants in Section \ref{sec-adapt}.  For simplicity, we shall present the notation, procedures, and theoretical results for one-sided hypotheses tests in Sections \ref{sec-notation}-\ref{sec-adapt}.
In Section \ref{sec-twosided}
we present the necessary modifications for two-sided hypotheses,
which turn out to be minimal. In Section \ref{sec-estthresholds} we suggest selection rules with
optimal properties.
%consider incorporating in our procedures the selection rules with
%two parameters, and suggest estimation methods for the optimal
%parameters of these procedures.
 In Sections \ref{sec-sim} and
\ref{sec-example} we present a simulation study and real data
examples, respectively. Conclusions are given in Section
\ref{sec-Discussion}. Lengthy proofs of theoretical results are in
the Appendix.

\subsection{Review of the plug-in estimator for estimating the fraction of nulls}\label{subsec-reviewplug-in}

 Let $\pi_0$ be the fraction of null hypotheses. \cite{schweder82} proposed estimating this fraction by $  \frac{\#\{
p-values>\lambda\}}{m(1-\lambda)},$ where $m$ is the number of
features and $\lambda \in (0,1)$.
 The slightly inflated  plug-in estimator  $$\hat \pi_0 =  \frac{\# \{ p-values>\lambda\}+1}{m(1-\lambda)}$$ has been incorporated into multiple testing procedures in recent years. %The adaptive procedures replace $m$ with $m\hat \pi_0$ in the multiple testing procedure.
  For independent $p$-values, \cite{storey2} proved that applying the BH procedure with $m\hat \pi_0$ instead of $m$ controls the FDR, and \cite{finner09} proved that applying Bonferroni with $m\hat \pi_0$ instead of $m$ controls the FWER.

Adaptive procedures in single studies have larger power gain over
non-adaptive procedures when the fraction of nulls, $\pi_0$, is
small. This is so because these procedures essentially apply the
original procedure at level $1/\hat \pi_0$ times the nominal level
to achieve FDR or FWER control at the nominal level.
\cite{finner09} showed in simulations that the power gain of using
$m\hat \pi_0$ instead of $m$ can be small when the fraction of nulls
is 60\%, but large when the fraction of nulls is 20\%. 

The plug-in estimator  is typically less conservative (smaller) the
larger $\lambda$ is.  This follows from Lemma 1 in
\cite{Dickhaus12}, that showed  that for a single study the
estimator is biased upwards, and that the bias is a decreasing
function of $\lambda$ if the cumulative distribution function (CDF)
of the non-null $p$-values is concave (if the $p$-values are based
on a test statistic whose density is eventually strictly decreasing,
then concavity will hold, at least for small $\lambda$).
\cite{yoav2} noted that the FDR of the BH procedure which
incorporates the plug-in estimator with $\lambda=0.5$ is sensitive
to deviations from the assumption of independence, and it may be
inflated above the nominal level under dependency.
\cite{Blanchard09} further noted that although under
equi-correlation among the test statistics using the plug-in
estimators does not control the FDR with $\lambda=0.5$, it does
control the FDR with $\lambda= q/(q+1+1/m) \approx q$.
\cite{Blanchard09} compared in simulations with dependent test
statistics the adaptive BH procedure using various estimators of the
fraction of nulls for single studies, including the plug-in
estimator with $\lambda \in \{0.05,0.5\}$. Their conclusion was that
the plug-in estimator with $\lambda = 0.05$ was superior to all
other estimators considered,  since it had the highest power overall
without inflating the FDR above the 0.05 nominal level.

\section{Notation, goal, and review for replicability analysis }\label{sec-notation}
 Consider a family of $m$ features examined in two independent studies. The effect of feature $j\in \{ 1,\ldots,m\}$ in study $i\in \{1,2\}$ is $\theta_{ij}$. Let $H_{ij}$ be the hypothesis indicator, so $H_{ij} = 0$ if $\theta_{ij} =  \theta_{ij}^0$, and $H_{ij} = 1$ if $\theta_{ij}> \theta_{ij}^0$.

Let $\vec H_j = (H_{1j}, H_{2j})$. The set of  possible states of
$\vec H_j$ is $\mathcal{H} = \{\vec{h} = (h_1,h_2): (0,0), (1,0),
(0,1), (1,1) \}.$ The
goal of inference is to discover as many features as possible with
$\vec H_j \notin \mathcal H^0$, where $\mathcal
H^0\subset\mathcal{H}.$ For replicability analysis, $\mathcal{H}^0 =
\mathcal{H}^0_{NR} =  \{(0,0), (0,1), (1,0)\}$. For a typical meta-analysis, $\mathcal{H}^0 =\{(0,0)\}$,  and the number of features with state $(0,0)$ can be much smaller than the number of features with states in $\mathcal{H}^0_{NR}$, see the example in Section \ref{subsec-big-example}.

We aim to discover as many  features with $\vec H_j = (1,1)$
 as possible, i.e., true
replicability claims, while controlling for false replicability
claims, i.e. replicability claims for features with $\vec H_j
\in\mathcal{H}^0_{NR} .$ Let $\mathcal R$ be the set of indices of
features with replicability claims. The FWER and FDR for
replicability analysis are defined as follows:
$$FWER = \textmd{Pr}\left(|\mathcal R \cap \{j: \vec H_j \in \mathcal{H}^0_{NR} \} |>0\right), \quad FDR = E\left(\frac{|\mathcal R \cap \{j: \vec H_j \in \mathcal{H}^0_{NR} \} |}{\max(|\mathcal R|,1 )} \right), $$
where $E(\cdot)$ is the expectation.

Our novel procedures first select promising features from each study solely based on the data of that study.
Let $\mathcal S_i$ be the index set of features selected in
study $i,$ for $i\in \{1,2\},$ and let $S_i= |\mathcal S_i|$ be
their number. The procedures proceed towards making replicability
claims only on the index set of features which are selected in both
studies, i.e.
$\mathcal{S}_1\cap \mathcal{S}_2.$ For example, selected sets may include all (or a subset of) features with two-sided $p$-values below $\alpha$. See Remark \ref{rem-selectiontype} for a discussion about the selection process. 

Let $P_i=(P_{i1}, \ldots,P_{im})$ be the $m$-dimensional random
vector of $p$-values of study $i\in\{1,2\},$ and $p_i=(p_{i1},
\ldots,p_{im})$ be its realization.
We shall assume the following condition is satisfied for $(P_1,
P_2)$:
\begin{definition}The studies satisfy the \emph{null independence-across-studies condition} if
for all $j$ with $\vec{H}_j\in \mathcal{H}^0_{NR}$, if $H_{1j}=0$
then $P_{1j}$ is independent of $P_2$, and if $H_{2j}=0$ then
$P_{2j}$ is independent of $P_1$.
\end{definition}
This condition is clearly satisfied if the two studies are
independent, but it  also allows the pairs $(P_{1j},
P_{2j})$ to be dependent for $\vec H_j\notin  \mathcal{H}^0_{NR}$.
Note moreover that this condition does not pose any restriction on
the joint distribution of $p$-values within each study.

We shall assess the evidence towards replicability by a quantity we call the $r$-value, introduced in \cite{Heller14b}, which is the adjusted $p$-value for replicability analysis.
In a single study, the adjusted $p$-value of a feature is the
smallest level (of FWER or FDR) at which it is discovered
\citep{wright92}. Similarly, for feature $j$,
the $r$-value is the smallest level (of FWER or FDR) at which
feature $j$ is declared replicable.

The simplest example of $p$-value adjustment for a single
study $i\in \{1,2\}$ is Bonferroni, with adjusted $p$-values
$p_{ij}^{adj-Bonf}=m p_{ij}, j=1,\ldots, m$. The BH adjusted
$p$-values build upon the Bonferroni adjusted $p$-values
\citep{reiner2}.
 The BH adjusted $p$-value for
feature $j$ is defined to be
$$  \min_{\{k:\, p_{ik}^{adj-Bonf}\geq p_{ij}^{adj-Bonf},\, k=1,\ldots,m \}} \frac{p_{ik}^{adj-Bonf}}{rank(p_{ik}^{adj-Bonf})},$$
where $rank(p_{ik}^{adj-Bonf})$ is the rank of the Bonferroni
adjusted $p$-value for feature $k$, with maximum rank for ties.
For two studies, we can for example define the Bonferroni-on-max
$r$-values to be $r_j^{Bonf-max}=m\max (p_{1j}, p_{2j}), j=1,\ldots,
m$. The BH-on-max $r$-values build upon the Bonferroni-on-max
$r$-values exactly as in single studies. The BH-on-max $r$-value for
feature $j$  is defined to be
$$  \min_{\{k: \,r_k^{Bonf-max}\geq r_j^{Bonf-max},\, k=1,\ldots,m \}} \frac{r_k^{Bonf-max}}{rank(r_k^{Bonf-max})},$$
where $rank(r_k^{Bonf-max})$ is the rank of the Bonferroni-on-max
adjusted $p$-value for feature $k$, with maximum rank for ties.
Claiming as replicable the findings of all features with BH-on-max
$r$-values at most $\alpha$ is equivalent to considering as
replicability claims the discoveries from applying the BH procedure
at level $\alpha$ on the maximum of the two studies $p$-values,
suggested in \cite{benjamini09}. In this work we introduce
$r$-values that are typically much smaller than the above-mentioned
$r$-values for features selected in both studies, with the same
theoretical guarantees upon rejection at level $\alpha$, and thus
preferred for replicability analysis of two studies.

\section{Replicability among the selected in each of two studies}\label{sec-non-adapt}
Let $c\in (0,1)$, with default value $c=0.5$, be the fraction of the
significance level ``dedicated" to study one. The Bonferroni
$r$-values are
 $$r^{Bonf}_j =
\max\left(\frac{S_2p_{1j}}{c}, \frac{S_1p_{2j}}{1-c}\right), \quad j
\in \mathcal S_1\cap \mathcal S_2.
$$

The FDR $r$-values build upon the Bonferroni $r$-values and are
necessarily smaller: \begin{align} r^{FDR}_j = \min_{\{i:\,
r^{Bonf}_i\geq r^{Bonf}_j,\, i \in \mathcal S_1\cap \mathcal S_2
\}} \frac{r^{Bonf}_i}{rank(r^{Bonf}_i)},\quad j \in \mathcal S_1\cap
\mathcal S_2.\label{r_FDR}\end{align} where $rank(r^{Bonf}_i)$ is
the rank of the Bonferroni $r$-value for feature $i \in \mathcal
S_1\cap \mathcal S_2$, with maximum rank for ties.

Declaring as replicated all features with Bonferroni $r$-values at most $\alpha$ controls the FWER at level $\alpha$, and declaring as replicated all features with FDR $r$-values at most $\alpha$ controls the FDR at level $\alpha$ under independence, see Section \ref{subsec-theoreticalproperties}.

The relation between the Bonferroni and FDR $r$-values is similar to
that of the adjusted Bonferroni and adjusted BH $p$-values described
in Section \ref{sec-notation}. For the features selected in both
studies, if less than half of the features are selected by each
study, it is easy to show that  FDR (Bonferroni) $r$-values given above, using
$c=0.5$, will be smaller than (1) the 
BH-on-max (Bonferroni-on-max) $r$-values described in Section \ref{sec-notation}, and (2) the $r$-values that correspond to the FDR-controlling symmetric procedure in \cite{Bogomolov13}, which
will be typically smaller than BH-on-max $r$-values but larger than FDR $r$-values in (\ref{r_FDR}) due to taking into account the multiplicity of all features considered.

\subsection{Theoretical properties}\label{subsec-theoreticalproperties}
Let $\alpha\in (0,1)$ be the level of control desired, e.g. $\alpha = 0.05$. Let $\alpha_1=c\alpha$ be the fraction of $\alpha$ for study one, e.g. $\alpha_1 = \alpha/2$.

The procedure that makes replicability claims for features with
Bonferroni $r$-values at most $\alpha$ is a special case of the
following more general procedure.

\begin{procedure}\label{proc-FWER}
FWER-replicability analysis on the selected features $\mathcal
S_1\cap \mathcal S_2$:
\begin{enumerate}
\item Apply a FWER controlling
procedure at level $\alpha_1$ on the set $\{p_{1j}, j\in
\mathcal{S}_2\},$ and let $\mathcal R_{1}$ be the set of indices of discovered features. Similarly, apply a FWER controlling procedure at level
$\alpha-\alpha_1$  on the set $\{p_{2j}, j\in \mathcal{S}_1\},$  and let
$\mathcal R_{2}$ be the set of indices of discovered features.
\item The set of indices of features with replicability claims is  $\mathcal R_{1}\cap \mathcal R_{2}$.
\end{enumerate}
\end{procedure}

When using Bonferroni in Procedure \ref{proc-FWER},
feature $j\in \mathcal S_1\cap \mathcal S_2$ is among the
discoveries if and only if $(p_{1j}, p_{2j})\leq (\alpha_1/S_2,
\quad (\alpha-\alpha_1)/S_1). $ Therefore,  claiming replicability for all features 
with Bonferroni $r$-values at most $\alpha$ is equivalent to
Procedure \ref{proc-FWER} using Bonferroni. 

\begin{theorem}\label{thm-fwer}
If the null independence-across-studies condition is satisfied, 
then Procedure \ref{proc-FWER} controls the FWER for replicability
analysis at level $\alpha$.
\end{theorem}
\begin{proof}
Let $V_1 =|\mathcal R_1 \cap \{j:H_{1j}=0\} | $ and $V_2 = |\mathcal R_2 \cap \{j:H_{2j}=0\}|$ be the number of true %elementary
null hypotheses rejected in study one and in study two,
respectively, by Procedure \ref{proc-FWER}. Then the FWER for replicability analysis is
\begin{eqnarray}
E(\textbf{I}[V_1+V_2>0])\leq E(E(\textbf{I}[V_1>0]|P_2))
+E(E(\textbf{I}[V_2>0]|P_1)). 
\nonumber
\end{eqnarray}
Clearly, $E(\textbf{I}[V_1>0]|P_2)\leq \alpha_1$ since $P_{1j}$  is
independent of $P_2$ for all $j$ with $H_{1j}=0$, and a FWER controlling
procedure is applied on $\{p_{1j}, j\in
\mathcal{S}_2\}.$ Similarly, $E(\textbf{I}[V_2>0]|P_1)\leq \alpha-\alpha_1$. It thus follows that the FWER for replicability analysis is at most $\alpha$. 
\end{proof}

The procedure that rejects the features with FDR $r$-values at most
$\alpha$ is equivalent to the following procedure, see Lemma
\ref{lem_fdr} for a proof.
\begin{procedure}\label{procfdrsym}
FDR-replicability analysis on the selected features $\mathcal
S_1\cap \mathcal S_2$:
\begin{enumerate}
\item Let $$R\triangleq\max\left\{r:
\sum_{j\in\mathcal S_1\cap \mathcal{S}_2}\textbf{I}\left[(p_{1j},
p_{2j})\leq\left(\frac{r\alpha_1}{S_2},
\frac{r(\alpha-\alpha_1)}{S_1}\right)\right] = r\right\}.$$
\item The set of indices with replicability claims is
\begin{align*}\mathcal R= \{j: (p_{1j},
p_{2j})\leq\left(\frac{R\alpha_1}{S_2},
\frac{R(\alpha-\alpha_1)}{S_1}\right), j \in \mathcal S_1\cap
\mathcal S_2\}.\end{align*}
\end{enumerate}
\end{procedure}

This procedure controls the FDR for replicability analysis at level
$\alpha$ as long as the selection rules by which the sets
$\mathcal{S}_1$ and $\mathcal{S}_2$ are selected are stable (this is a very lenient requirement, see \cite{Bogomolov13} for examples).

\begin{definition}\citep{Bogomolov13} A stable selection rule satisfies the following condition: for any
selected feature, changing its $p$-value so that the feature is
still selected while all other $p$-values are held fixed, will not
change the set of selected features.\end{definition}

 \begin{theorem}\label{indep}
If the null independence-across-studies condition is satisfied, and
the selection rules by which the sets
$\mathcal{S}_1$ and $\mathcal{S}_2$ are selected are stable,
 then Procedure \ref{procfdrsym} controls the FDR for replicability
analysis at level $\alpha$  if one of the following items is
satisfied: 
\begin{enumerate}
\item[(1)]  The $p$-values from true null hypotheses
within each study are each independent of all other $p$-values.
\item[(2)] Arbitrary dependence among the $p$-values within each study, when $S_i$ in Procedure \ref{procfdrsym} is replaced by $S_i\sum_{k=1}^{S_i}1/k$, for $i=1,2$.
\end{enumerate}

\end{theorem}
See Appendix \ref{app-thm-fdr-indep} for a proof. 

\begin{remark}
The FDR $r$-values  for the procedure that is valid for arbitrary
dependence, denoted by $\tilde{r}_j^{FDR}, j\in \mathcal{S}_1\cap \mathcal{S}_2$, are computed using formula (\ref{r_FDR}) where the
Bonferroni $r$-values $r_j^{Bonf}$ are replaced by
 \begin{align}\tilde{r}_j =
\max\left(\frac{(\sum_{i=1}^{S_2}1/i)S_2p_{1j}}{c},
\frac{(\sum_{i=1}^{S_1}1/i)S_1p_{2j}}{1-c}\right), \quad j \in
\mathcal S_1\cap \mathcal S_2.\label{rtilda}
\end{align}

\end{remark}

\begin{remark}\label{rem-selectiontype}
An intuitive approach towards replicability may be to apply a multiple testing procedure on each study separately, with discovery sets $\mathcal D_1$ and $\mathcal D_2$ in study one and two, respectively, and then claim replicability on the set $\mathcal D_1\cap \mathcal D_2$. However, even if the multiple testing procedure has guaranteed FDR control at  level $\alpha$, it is easy to construct examples where the expected fraction of false replicability claims in  $\mathcal D_1\cap \mathcal D_2$ will be far larger than $\alpha$. An extreme example is the following: half of the  features have  $\vec H_j = (1,0)$, the remaining half have  $\vec H_j = (0,1)$, and the signal is very strong. Then in study one all features with  $\vec H_j = (1,0)$ and few features with $\vec H_j = (0,1)$ will be discovered, and in study two all features with  $\vec H_j = (0,1)$ and few features with $\vec H_j = (1,0)$ will be discovered, resulting in a non-empty set $\mathcal D_1\cap \mathcal D_2$ which contains only false replicability claims. Interestingly, if the multiple testing procedure is Bonferroni at level $\alpha$, then the FWER on replicability claims of the set $\mathcal D_1\cap \mathcal D_2$  is at most $\alpha$. However, this procedure (which can be viewed as Bonferroni on the maximum of the two study $p$-values) can be far more conservative than our suggested Bonferroni-type procedure. If we select in each study separately all features with $p$-values below $\alpha/2$, resulting in selection sets $\mathcal S_1$ and $\mathcal S_2$ in study one and two, respectively, then using our Bonferroni-type procedure we claim replicability for features with $(p_{1j}, p_{2j})\leq (\alpha/(2S_2),\alpha/(2S_1))$. Our discovery thresholds, $ (\alpha/(2S_2),\alpha/(2S_1))$, are both larger than $\alpha/m$ as long as the number of features selected by each study is less than half, and thus can lead to more replicability claims  with FWER control at level $\alpha$. 
\end{remark}

\section{Incorporating the plug-in estimates}\label{sec-adapt}
When the non-null hypotheses are mostly non-null in both studies,
i.e.,  there are more features with $\vec H_j = (1,1)$ than with
$\vec H_j = (1,0)$ or  $\vec H_j = (0,1)$, then the non-adaptive
procedures for replicability analysis may be over conservative. The
conservativeness follows from the fact that the fraction of null
hypotheses in one study among the selected in the other study is small.  %In
%$\mathcal{S}_1$  the hypotheses with $\vec H_j
%\in \{(0,0), (0,1\}$ have a smaller chance of being among the
%selected in  $\mathcal{S}_1$,
The set $\mathcal{S}_1$ is more likely to contain hypotheses with
$\vec H_j \in \{(1,0), (1,1)\}$ than hypotheses with $\vec H_j \in
\{(0,0), (0,1)\},$ and therefore the fraction of true null
hypotheses in study two among the selected in study one, i.e.,
$\sum_{j \in \mathcal{S}_1} (1-H_{2j})/S_1$, may be much smaller
than one (especially if there are more features with $\vec H_j =
(1,1)$ than with $\vec H_j = (1,0)$). Similarly, the fraction of
true null hypotheses in study one among the selected based on study
two, i.e., $\sum_{j \in \mathcal{S}_2} (1-H_{1j})/S_2$, may be much
smaller than one.

The non-adaptive procedures for replicability analysis in Section \ref{sec-non-adapt}  control the error-rates at levels that are conservative by the expectation of these fractions.
Procedures \ref{proc-FWER} using Bonferroni and \ref{procfdrsym}
control the FWER and FDR, respectively, at level which is at most
\begin{equation}
\alpha_1E\left(\frac{\sum_{j \in \mathcal{S}_2}
(1-H_{1j})}{S_2}\right)+(\alpha-\alpha_1)E\left(\frac{\sum_{j\in
\mathcal{S}_1} (1-H_{2j})}{S_1}\right), \nonumber
\end{equation}
 which can be much
smaller than $\alpha$ if  the above expectations are far smaller
than one. This upper bound follows for FWER since an upper bound
for the FWER of a  Bonferroni procedure is the desired level times
the fraction of null hypotheses in the family tested, and  for the FDR   from the proof of item 1 of Theorem \ref{indep}.

 We therefore suggest adaptive variants, that first estimate the expected fractions of true null hypotheses among the selected. We use the slightly inflated plug-in estimators (reviewed in Section \ref{subsec-reviewplug-in}):
\begin{align}\label{eq-adaptiveestimates}
\prone = \frac{1+\sum_{j\in
\mathcal{S}_{2,\lambda}}\textbf{I}(P_{1j}>\lambda)}{S_{2,\lambda}(1-\lambda)};\,\,\,
\prtwo = \frac{1+\sum_{j\in
\mathcal{S}_{1,\lambda}}\textbf{I}(P_{2j}>\lambda)}{S_{1,\lambda}(1-\lambda)},
 \end{align}
 where $0<\lambda<1$ is a fixed parameter, $\mathcal{S}_{i,\lambda} = \mathcal{S}_i\cap \{j:P_{ij}\leq \lambda\}$,
and  $S_{i,\lambda} = |\mathcal{S}_{i,\lambda}|$, for $i=1,2$.
 Although $\prone$ and $\prtwo$ depend on the tuning parameter $\lambda,$ we suppress the dependence of the estimates on
 $\lambda$ for ease of notation.

 The adaptive Bonferroni $r$-values  for fixed $c = \alpha_1/\alpha$
are: $$r^{adaptBonf}_j = \max\left(\frac{\prone
S_{2,\lambda}p_{1j}}{c}, \frac{\prtwo
S_{1,\lambda}p_{2j}}{1-c}\right), \quad j \in \mathcal
S_{1,\lambda}\cap \mathcal S_{2,\lambda}.$$ As in
Section \ref{sec-non-adapt}, the adaptive FDR $r$-values  build upon the
adaptive Bonferroni $r$-values:
$$ r^{adaptFDR}_j = \min_{\{i:\, r^{adaptBonf}_i\geq r^{adaptBonf}_j,\, i \in \mathcal{S}_{1,\lambda}\cap\mathcal{S}_{2,\lambda} \}} \frac{r^{adaptBonf}_i }{rank(r^{adaptBonf}_i)}, \quad j \in
\mathcal S_{1,\lambda}\cap \mathcal S_{2,\lambda}$$ where
$rank(r^{adaptBonf}_i)$ is the rank of the adaptive Bonferroni
$r$-value for feature $i \in \mathcal{S}_{1,\lambda}\cap
\mathcal{S}_{2,\lambda}$, with maximum rank for ties. Declaring as replicated all features with  adaptive Bonferroni/FDR $r$-values at most $\alpha$
controls the FWER/FDR for replicability analysis at level $\alpha$ under independence, see
Section \ref{sec-adapt-theoreticalproperties}.

The non-adaptive procedures in Section \ref{sec-non-adapt} only
require as input $\{p_{1j}: j \in \mathcal S_1 \}$ and $\{p_{2j}: j \in \mathcal S_2 \}$. However, if $\{p_{1j}: j \in \mathcal S_1\cup  \mathcal S_2 \}$ and $\{p_{2j}: j \in \mathcal S_1\cup \mathcal S_2 \}$ are available,  then  the adaptive procedures with $\lambda = \alpha$ are  attractive alternatives with better power, as demonstrated in  our simulations detailed in Section \ref{sec-sim}.

\subsection{Theoretical properties}\label{sec-adapt-theoreticalproperties}
The following Procedure \ref{proc-bonferroniadapt}  is equivalent to declaring as replicated all
features with Bonferroni adaptive $r$-values at most $\alpha$.
\begin{procedure}\label{proc-bonferroniadapt}
Adaptive-Bonferroni-replicability analysis on $\{(p_{1j},p_{2j}): j \in \mathcal S_1\cup S_2 \}$ with input parameter $\lambda$:
\begin{enumerate}
\item Compute $\prone, \prtwo$ and $S_{1,\lambda}, S_{2,\lambda}$.
\item  Let $\mathcal R_{1} = \{j\in \mathcal{S}_{1,\lambda}: p_{1j}\leq \alpha_1/(S_{2,\lambda}\prone)\}$ and $\mathcal R_{2} = \{j\in \mathcal{S}_{2,\lambda}: p_{2j}\leq (\alpha-\alpha_1)/(S_{1,\lambda}\prtwo)\}$
be the sets of indices of  features discovered in studies one and
two, respectively.
\item The set of indices of features with replicability claims is  $\mathcal R_{1}\cap \mathcal R_{2}$.
\end{enumerate}
\end{procedure}

\begin{theorem}\label{thm-bonferroniadapt}
If the null independence-across-studies condition is satisfied, and
the $p$-values from true null hypotheses within each study are
jointly
independent, 
then Procedure \ref{proc-bonferroniadapt} controls
the FWER for replicability analysis at level $\alpha$.
\end{theorem}
\begin{proof}
It is enough to prove that $E(\textbf{I}[V_1>0]|P_2)\leq \alpha_1$
and $E(\textbf{I}[V_2>0]|P_1)\leq \alpha-\alpha_1,$ as we showed in
the proof of Theorem \ref{thm-fwer}. These inequalities essentially
follow from the fact that the Bonferroni plug-in procedure controls the FWER
 \citep{finner09}. We will only show that
$E(\textbf{I}[V_1>0]|P_2)\leq \alpha_1$, since the proof that
$E(\textbf{I}[V_2>0]|P_1)\leq \alpha-\alpha_1$ is similar. We shall
use the fact that
\begin{equation}
\prone\geq
\frac{1+\sum_{j\in\mathcal{S}_{2,\lambda}}(1-H_{1j})\textbf{I}(P_{1j}>\lambda)}{S_{2,\lambda}(1-\lambda)}.
 \label{eq-lowerboundfractionnulls}
\end{equation}

\begin{eqnarray}
&&E(\textbf{I}[V_1>0]|P_2) =  \textmd{Pr}\left(\sum_{i \in \mathcal{S}_{2,\lambda}}(1-H_{1i})\textbf{I}[i\in S_{1,\lambda}, P_{1i}\leq \alpha_1/(S_{2,\lambda}\prone)]>0| P_2\right)  \nonumber \\
&& \leq\sum_{i \in \mathcal{S}_{2,\lambda}}(1-H_{1i}) \textmd{Pr}(P_{1i}\leq \min(\lambda, \alpha_1/S_{2,\lambda}\prone)| P_2) \label{eq-bonfadapt1}\\
&&\leq\sum_{i \in \mathcal{S}_{2,\lambda}}(1-H_{1i})
\textmd{Pr}\left(P_{1i}\leq\min\left(\lambda,\frac{\alpha_1}{\left(\frac{1+\sum_{j\in\mathcal{S}_{2,\lambda}}(1-H_{1j})\textbf{I}(P_{1j}>\lambda)}{1-\lambda}\right)}\right)|P_2\right)\label{eq-bonfadapt2}
\\&& = \sum_{i \in \mathcal{S}_{2,\lambda}}(1-H_{1i}) \textmd{Pr}\left(P_{1i}\leq \min\left(\lambda, \frac{\alpha_1}{\left(\frac{1+\sum_{j\in \mathcal{S}_{2,\lambda}, j\neq i}(1-H_{1j})\textbf{I}(P_{1j}>\lambda)}{1-\lambda}\right)}\right)| P_2\right) \notag\\
&& \leq \sum_{i \in \mathcal{S}_{2,\lambda}}(1-H_{1i}) \alpha_1 E\left(1/\left(\frac{1+\sum_{j\in \mathcal{S}_{2,\lambda}, j\neq i}(1-H_{1j})\textbf{I}(P_{1j}>\lambda)}{1-\lambda}\right)| P_2 \right) \label{eq-bonfadapt4}\\
&& \leq \sum_{i \in \mathcal{S}_{2,\lambda}}
(1-H_{1i})\alpha_1/\sum_{j \in \mathcal{S}_{2,\lambda}} (1-H_{1j}) =
\alpha_1. \label{eq-bonfadapt5}
\end{eqnarray}
Inequality (\ref{eq-bonfadapt1}) follows from the Bonferroni
inequality, and
 inequality (\ref{eq-bonfadapt2}) follows from (\ref{eq-lowerboundfractionnulls}). 
 Inequality (\ref{eq-bonfadapt4}) follows from the facts that for $i$ with $H_{1i}=0$, (a) $P_{1i}$ is independent of all null $p$-values from study one and from all $p$-values from study two, and (b) $\textmd{Pr}(P_{1i}\leq x)\leq x$ for all $x\in[0,1].$ 
 Inequality (\ref{eq-bonfadapt5}) follows by
applying Lemma 1 in \cite{yoav2}, which states that if $Y\sim
B(k-1,p)$ then $E(1/(Y+1))<1/(kp)$,   to $Y = \sum_{j\in
\mathcal{S}_{2,\lambda}, j\neq
i}(1-H_{1j})\textbf{I}(P_{1j}>\lambda)$, which is distributed
$B(\sum_{j \in \mathcal{S}_{2,\lambda}, j \neq i}(1-H_{1j}),
1-\lambda)$ if the null $p$-values within each study are uniformly
distributed. It is easy to show, using similar arguments, that
inequality (\ref{eq-bonfadapt5}) remains true when the null
$p$-values are stochastically larger than uniform.
\end{proof}

Declaring as replicated all features with adaptive FDR $r$-values at most
$\alpha$ is equivalent to Procedure \ref{procfdrsym} where $S_1$ and
$S_2$ are replaced by $S_{1,\lambda}\prtwo$ and
$S_{2,\lambda}\prone$ respectively, and $\mathcal{S}_1\cap
\mathcal{S}_2$ is replaced by $\mathcal{S}_{1,\lambda}\cap
\mathcal{S}_{2,\lambda},$ see Lemma \ref{lem_fdr} for a proof.

\begin{theorem}\label{thm-proc-fdr-adaptive}
If the null independence-across-studies condition holds, the
$p$-values corresponding to true null hypotheses are each
independent of all the other $p$-values, and the selection rules by
which the sets $\mathcal{S}_1$ and $\mathcal{S}_2$ are selected are
stable, then declaring as replicated all features with adaptive FDR $r$-values at most $
\alpha$ controls the
FDR for replicability analysis at level $\alpha$.
\end{theorem}
See Appendix \ref{app-thm-fdr-indep} for a proof.

\section{Directional replicability analysis for two-sided alternatives}\label{sec-twosided}

So far we have considered one sided alternatives. If a two-sided alternative is considered for each feature in each study, and the aim is to discover the features that have replicated effect in the same direction in both studies, the following simple modifications are necessary.

For feature $j\in \{1,\ldots,m\}$, the left- and right- sided
$p$-values for study $i\in \{1,2\}$ are denoted by $p^L_{ij}$ and
$p^R_{ij}$, respectively. For continuous test statistics, $p^R_{ij}
= 1-p^L_{ij}$. 

For directional replicability analysis, the selection step has to be modified to include also the selection of the direction of testing. The set of
features selected is the
subset of features that are selected from both studies, for which the
direction of the alternative with the smallest one-sided $p$-value
is the same for both studies, i.e.,
$$\mathcal{S}\triangleq\mathcal{S}_1\cap\mathcal{S}_2\cap\left(\{j: \max(p_{1j}^R,
p_{2j}^R)<0.5\}\cup\{j: \max(p_{1j}^L, p_{2j}^L)<0.5\} \right).$$ In
addition, define for $j\in \mathcal{S}_1\cup\mathcal{S}_2,$
\begin{equation*}
p'_{1j} = \left\{
\begin{array}{rl}
p_{1j}^L & \text{if }  p_{2j}^L< p_{2j}^R,\\
p_{1j}^R & \text{if }  p_{2j}^L> p_{2j}^R,
\end{array} \right.
\end{equation*}

\begin{equation*}
p'_{2j} = \left\{
\begin{array}{rl}
p_{2j}^L & \text{if }  p_{1j}^L< p_{1j}^R,\\
p_{2j}^R & \text{if }  p_{1j}^L> p_{1j}^R.
\end{array} \right.
\end{equation*}

The Bonferroni and FDR $r$-values are computed for features in
$\mathcal{S}$ using the formulae given in Sections
\ref{sec-non-adapt} and \ref{sec-adapt} (where 
$\mathcal{S}_{1,\lambda}$ and $ \mathcal{S}_{2,\lambda}$ are the selected sets in Section \ref{sec-adapt}), with the following
modifications: the set $\mathcal{S}_1\cap \mathcal{S}_2$ is replaced
by $\mathcal{S}$, and $p_{1j}$ and $p_{2j}$ are replaced by
$p'_{1j}$ and $p'_{2j}$ for $j\in
\mathcal{S}_1\cup\mathcal{S}_2.$

As in Sections \ref{sec-non-adapt} and \ref{sec-adapt},  features with  $r$-values at most $\alpha$ are declared
as replicated at level $\alpha$, in the direction selected.  
The corresponding procedures remain valid, with the
theoretical guarantees of directional FWER and FDR control for
replicability analysis on the modified selected set above, despite
the fact that the direction of the alternative for establishing
replicability was not known in advance. This is remarkable, since it
means that there is no additional penalty, beyond the penalty for
selection used already in the above procedures, for the fact that
the direction for establishing replicability is also decided upon
selection. The proofs are similar to the proofs provided for one-sided hypotheses and are therefore omitted.

\section{Estimating the selection thresholds}\label{sec-estthresholds}
When the full data for both studies is available, we first need to
select the promising features from each study based on the data in
this study. If the selection is based on $p$-values, then our first
step will include selecting the features with $p$-values below
thresholds $t_1$ and $t_2$ for studies one and two, respectively. The
thresholds for selection, $(t_1,t_2)\in (0,1]^2$, affect power: if
$(t_1,t_2)$ are too low, features with $\vec H_j\notin
\mathcal{H}^0_{NR}$ may not be considered for replicability even if
they have a chance of being discovered upon selection, thus
resulting in power loss; if $(t_1,t_2)$ are too high,  too many
features with $\vec H_j\in \mathcal{H}^0_{NR}$ will be considered
for replicability making it difficult to discover the true
replicated findings, thus resulting in power loss.

We suggest automated methods for choosing $(t_1,t_2)$, based on
$(p_1,p_2)$ and the level of FWER or FDR control desired,  which are
based on the following principle: choose the values $(t_1,t_2)$  so
that the set of discovered features coincides with the set of
selected features.  We show in simulations in Section \ref{sec-sim}
that  data-dependent thresholds
 may lead to more powerful procedures than procedures with a-priori fixed
thresholds.

Let $\mathcal{S}_i(t_i)=\{j:p_{ij}\leq t_i\}$ be the index set of
selected features from study $i,$ for $i\in\{1,2\}.$ We suggest the
selection thresholds $(t^*_1,t^*_2)$ that solve the two equations
 \begin{equation}\label{eq-nonlin}
t_1 = \frac{\alpha_1}{|\mathcal{S}_2(t_2)|};  \quad t_2 =
\frac{\alpha-\alpha_1}{|\mathcal{S}_1(t_1)|},
\end{equation}
for Procedure \ref{proc-FWER} using Bonferroni, and the selection
thresholds $(t^*_1,t^*_2)$ that solve the two equations
 \begin{equation}\label{sel-fwer-adapt}
t_1 = \frac{\alpha_1}{|\mathcal{S}_{2,\lambda}(t_2)|\prone(t_2)};
\quad t_2 = \frac{\alpha-\alpha_1}{|\mathcal{S}_{1,
\lambda}(t_1)|\prtwo(t_1)},
\end{equation}
for the adaptive Procedure \ref{proc-bonferroniadapt} for FWER
control, where $\prone(t_2)$ and $\prtwo(t_1)$ are the estimators
defined in (\ref{eq-adaptiveestimates}) with
$\mathcal{S}_1=\mathcal{S}_{1,\lambda}(t_1)=\{j:P_{1j}\leq\min(\lambda,
t_1)\}$ and
$\mathcal{S}_2=\mathcal{S}_{2,\lambda}(t_2)=\{j:P_{2j}\leq\min(\lambda,
t_2)\}$. We show in Appendix \ref{app-estthresholdsTheoretical}
 that these choices are not
dominated by any other choices (i.e., there do not exist other
choices $(t_1,t_2)$ that result in larger rejection thresholds for
the $p$-values in both studies).

Similarly, we suggest the selection
thresholds $(t^*_1,t^*_2)$ that solve the two equations
\begin{equation}
t_1=\frac{|\mathcal{S}_1(t_1)\cap
\mathcal{S}_2(t_2)|\alpha_1}{|\mathcal{S}_2(t_2)|}; \quad
t_2=\frac{|\mathcal{S}_1(t_1)\cap
\mathcal{S}_2(t_2)|(\alpha-\alpha_1)}{|\mathcal{S}_1(t_1)|},\label{sel-fdr}
\end{equation}
for Procedure \ref{procfdrsym} for FDR control, and the selection
thresholds $(t^*_1,t^*_2)$ that solve the two equations
\begin{align}
&t_1=\frac{|\mathcal{S}_{1,\lambda}(t_1)\cap \mathcal{S}_{2,\lambda}(t_2)|\alpha_1}{|\mathcal{S}_{2,\lambda}(t_2)|\prone(t_2)},\notag\\
&t_2=\frac{|\mathcal{S}_{1,\lambda}(t_1)\cap
\mathcal{S}_{2,\lambda}(t_2)|(\alpha-\alpha_1)}{|\mathcal{S}_{1,\lambda}(t_1)|\prtwo(t_1)}.\label{sel-fdr-adapt}
\end{align}
for the adaptive FDR-controlling procedure in Section
\ref{sec-adapt}.

If the solution does not exist, no replicability claims are made.
There may be more than one solution to the equations
(\ref{eq-nonlin}) - (\ref{sel-fdr-adapt}). 
In our simulations and real data examples, we set as  $(t_1^*, t_2^*)$ the first solution outputted from the  algorithm used for solving the system of non-linear
equations.  We show in simulations that  using data-dependent thresholds $(t_1^*,t_2^*)$ results in  power close to the power using the optimal (yet unknown) fixed thresholds
$t_1=t_2$, and that  the nominal level of FWER/FDR is maintained under independence as well as under dependence  as long as we  use $\lambda=\alpha$ for the adaptive procedures.
 We prove in Appendix \ref{app-estthresholdsTheoretical} that  the nominal level of the non-adaptive procedures is controlled even though the selection thresholds $(t^*_1,t^*_2)$ are data-dependent, if the $p$-values are exchangeable under the null.

\section{Simulations}\label{sec-sim}
We define the configuration vector  $\vec f = (f_{00},f_{10},f_{01},f_{11})$, where $f_{lk}=\sum_{j=1}^mI[\vec{H}_j=(l,k)]/m,$ the proportion of features with state $(l,k)$, for $(l,k)\in\{0,1\}$. Given $\vec f$, measurements for $mf_{lk}$ features,  were generated from
$N(\mu_l,1)$ for study one, and $N(\mu_k,1)$ for study two, where
$\mu_0=0$ and $\mu_1=\mu>0$. One-sided $p$-values were computed for
each feature in each study. We varied $\vec f$ and $\mu \in
\{2,2.5,\ldots,6\}$ across simulations. We also examined the effect
of dependence within each study on the suggested procedures, by
allowing for equi-correlated test statistics within each study, with
correlation $\rho = \{0,0.25,0.5,0.75, 0.95\}$. Specifically, the
noise for feature $j$ in study $i\in \{1,2\}$ was
$e_{ij}=\sqrt{\rho}Z_{i0}+\sqrt{1-\rho}Z_{ij}$, where $\{Z_{ij}:
i=1,2, j=1,\ldots,m\}$ are independent identically distributed
$N(0,1)$ and $Z_{i0}$ is $N(0,1)$ random variable that is
independent of $\{Z_{ij}: i=1,2, j=1,\ldots,m\}$. The $p$-value for
feature $j$ in study $i$ was $1-\Phi(\mu_{ij}+e_{ij})$, where
$\Phi(\cdot)$ is the CDF of the standard normal distribution, and
$\mu_{ij}$ is the expectation for the signal of feature $j$ in study
$i$.

Our goal in this simulation was three-fold: First, to show the advantage of the adaptive procedures over the non-adaptive procedures for replicability analysis; Second, to examine the behaviour of the adaptive procedures when the test statistics are dependent within studies; Third, to compare the novel procedures with alternatives suggested in the literature.
The power, FWER, and FDR for the replicability analysis procedures considered were estimated based on 5000 simulated datasets.

\subsection{Results for FWER controlling procedures}\label{sec-fwersim}
We considered the following novel procedures:
Bonferroni-replicability with fixed or with data-dependent selection
thresholds $(t_1,t_2)$, adaptive-Bonferroni-replicability with
$\lambda \in \{0.05,0.5\}$ and with fixed or with data-dependent
$(t_1,t_2)$. These procedures were compared to an oracle procedure
with data-dependent thresholds (oracle-Bonferroni-replicability),
that knows $\sum_{j \in \mathcal S_2(t_2)}(1-H_{1j})$ and $\sum_{j
\in \mathcal S_1(t_1)}(1-H_{2j})$ and therefore rejects a feature
$j\in \mathcal{S}_1(t_1)\cap\mathcal{S}_2(t_2)$ if and only if
$p_{1j}\leq \alpha_1/\sum_{j \in \mathcal S_2(t_2)}(1-H_{1j})$ and
$p_{2j}\leq (\alpha-\alpha_1)/\sum_{j \in \mathcal
S_1(t_1)}(1-H_{2j})$. In addition, two procedures based on the
maximum of the two studies $p$-values were considered: the procedure
that declares as replicated all features with
$\max(p_{1i},p_{2i})\leq \alpha/m$ (Max), and the equivalent oracle
that knows $|\{j: \vec H_j\in \mathcal H_{NR}^0 \}|$ and therefore
declares as replicated all features with $\max(p_{1i},p_{2i})\leq
\alpha/|\{j: \vec H_j\in \mathcal H_{NR}^0 \}|$ (oracle Max). Note
that the oracle Max procedure controls the FWER for replicability analysis at the nominal level
$\alpha$ since the FWER is at most $\sum_{\{i: \vec H_i  \in
\mathcal H_{NR}^0 \}}\textmd{Pr}(\max(p_{1i},p_{2i})\leq
\alpha/|\{j: \vec H_j\in \mathcal H_{NR}^0 \}|\}\leq \alpha$.

Figure \ref{fig:fwerFixedThresholds}  shows the power for various
fixed selection thresholds $t_1=t_2=t$. There is a
clear gain from adaptivity since the power curves for the adaptive
procedures are above those for the non-adaptive procedures, for the
same fixed threshold $t.$ The gain
from adaptivity is larger as the difference between
$f_{11}$ and $f_{10} =f_{01}$ is larger: while in the last two rows
(where $f_{10} =f_{01}<f_{11}$) the power advantage can be greater than 10\%, in the first row (where $f_{10} =f_{01}=0.1, f_{11}=0.05$)
there is almost no power advantage. The  choice of  $t$ matters, and
 the power of the procedures with data-dependent
thresholds $(t_1^*, t_2^*)$ is close to the power of the procedures
with the best possible fixed threshold $t.$  

Figure \ref{fig-fwerproc} shows the power and FWER versus $\mu$
under independence (columns 1 and 2) and under equi-correlation of
the test statistics with $\rho=0.25$ (columns 3 and 4).  
The novel procedures are clearly superior to the Max and Oracle Max procedures, the adaptive procedures are superior to the non-adaptive variants, and the power of the adaptive procedures with  data-dependent thresholds is close to that of the oracle Bonferroni procedure. 
The adaptive procedures with $\lambda=0.05$
and  $\lambda=0.5$ have similar power, but the FWER with
$\lambda=0.05$ is controlled in all dependence settings while the FWER
with $\lambda=0.5$ is above 0.1 in all but the last dependence
setting.  Our results concur with the results of \cite{Blanchard09} for single
studies, that the preferred parameter is $\lambda=0.05$. The
adaptive procedure with $\lambda = 0.05$ and data-dependent
selection thresholds is clearly superior to the two adaptive
procedures with fixed selection thresholds of $t_1=t_2=0.025$ or
$t_1=t_2=0.049$. We thus recommend the
adaptive-Bonferroni-replicability procedure with $\lambda =0.05$ and
data-dependent selection thresholds.

\begin{figure}[htbp]
    \centering
\includegraphics[width =17.5cm,height = 22.5cm]{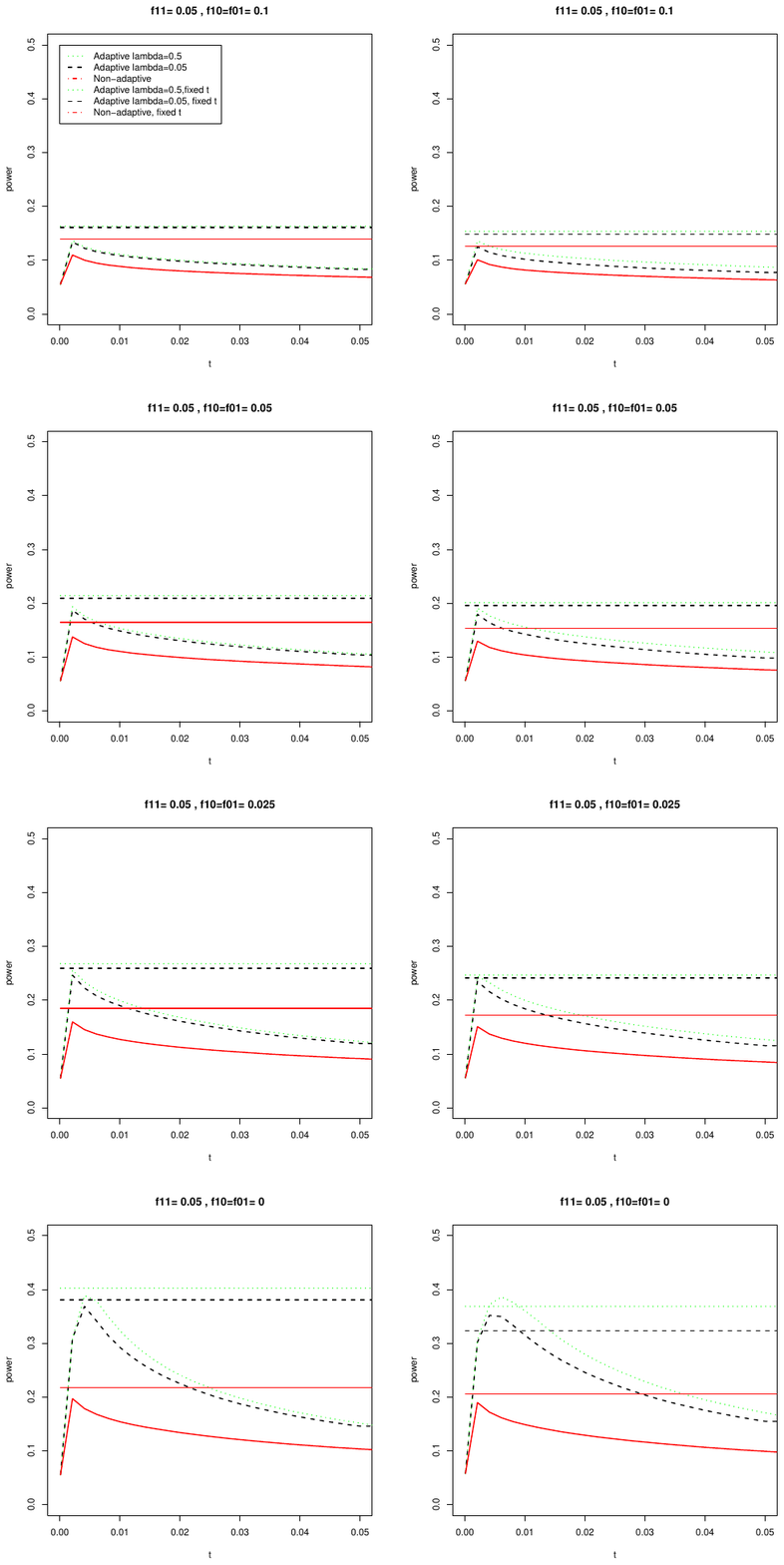}
\caption{ Column 1: Independence setting; Columns 2: Equicorrelation
with $\rho=0.25$. The power versus fixed threshold $t$ is shown for
$\mu = 3$ for the adaptive-Bonferroni-replicability procedure
(dashed black with $\lambda =0.05$ and dotted green with $\lambda =
0.5$) and non-adaptive Bonferroni-replicability procedure (solid
red), along with the power of these procedures with data-dependent
thresholds. In all settings $m=1000$, $\alpha=0.05, \alpha_1=0.025$.
}\label{fig:fwerFixedThresholds}
\end{figure}

\begin{center}
\begin{figure}[htbp]
    \centering
  \includegraphics[width =17.5cm,height = 22.5cm]{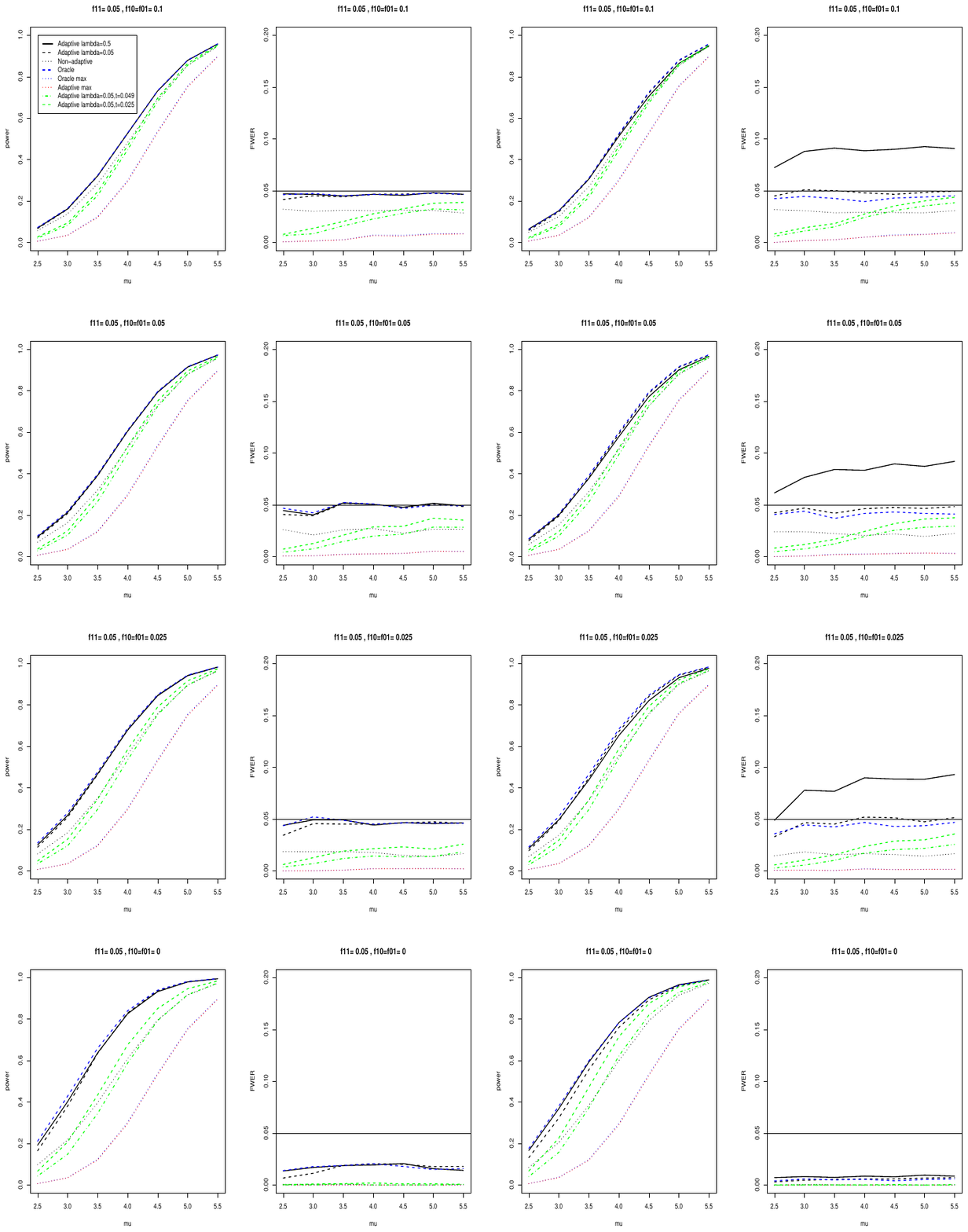}
    \caption{\scriptsize{Columns 1 and 2: Independence setting; Columns 3 and 4: Equi-correlation with $\rho=0.25$.
    Power and FWER versus $\mu$ for the adaptive-Bonferroni-replicability procedure with data-dependent $(t_1,t_2)$
    with $\lambda=0.5$ (solid black ) and with $\lambda=0.05$ (dashed black); Bonferroni-replicability procedure with data-dependent $(t_1,t_2)$
    (dotted black); the oracle that knows which hypotheses are null in one study among the selected from the other study (dashed blue);
    oracle Max (dotted blue)  and Max (dotted red); adaptive-Bonferroni-replicability with fixed $\lambda=0.05$ and fixed $t_1=t_2=0.049$
    (dash-dot green) and fixed $t_1=t_2=0.025$ (dash green).   In all settings $m=1000$, $\alpha=0.05, \alpha_1=0.025$. %The power and FWER  wer estimated from 5000 simulated datasets.
    }}
    \label{fig-fwerproc}
\end{figure}
\end{center}

\subsection{Results for FDR controlling procedures}\label{sec-fdrsim}
We considered the following novel procedures for replicability
analysis with $\alpha=0.05, \alpha_1 = 0.025$:
Non-adaptive-FDR-replicability with fixed or data-dependent
$(t_1,t_2)$; adaptive-FDR-replicability with $\lambda \in \{0.05,
0.5\}$ and fixed or data-dependent $(t_1,t_2)$.

\cite{Heller14} introduced the oracle Bayes procedure (oracleBayes), %that knows for each $j$ whether it belongs to $I_{00}, I_{01}, I_{10}$, or $I_{11}$, as well as the null and non-null densities of the $z$-scores $z_{ij} = \Phi^{-1}(1-p_{ij})$, for $i=1,2$ and $j=1,\ldots,m$,
and showed that it has the largest rejection region while
controlling the Bayes FDR. When  $m$ is large and the data is
generated from the mixture model, the Bayes FDR coincides with the
frequentist FDR, so oracle Bayes is optimal for FDR control.
We considered this oracle procedure for comparison with our novel procedures.
The difference in power between the oracle Bayes and the novel frequentist procedures 
shows how much worse our procedures, which
make no mixture-model assumptions,  are from the (best yet unknown in practice) oracle
procedure, which assumes the mixture model and needs as input its parameters.
 In addition, the following three procedures were considered: the empirical Bayes procedure (eBayes),
 as implemented in the R package \emph{repfdr} \citep{Heller14c},  which estimates the Bayes FDR and rejects the features with estimated Bayes FDR below $\alpha$,
  see \cite{Heller14} for details; the oracle BH on $\{\max(p_{1i}, p_{2i}): i=1,\ldots, m\}$ (oracleMax); and the adaptive BH on $\{\max(p_{1i}, p_{2i}): i=1,\ldots, m\}$ (adaptiveMax).
Specifics about oracleMax and adaptiveMax follow. Applying the BH on $\{\max(p_{1i}, p_{2i}): i=1,\ldots, m\}$ at level $x$,
it is easy to show that the FDR level for independent features is at most $f_{00}x^2+(1-f_{00}-f_{11})x$.
Therefore, the oracleMax procedure uses level $x$, which is the solution to $f_{00}x^2+(1-f_{00}-f_{11})x=0.05$,
and the adaptiveMax procedure uses level  $x$, which is the solution to $\hat f_{00}x^2+(1-\hat f_{00}-\hat f_{11})x=0.05$, where $\hat f_{00}$ and $\hat f_{11}$ are the estimated mixture fractions computed using the R package \emph{repfdr}.

Figure \ref{fig:fdrFixedThresholds}  shows the power of novel
procedures for various fixed selection thresholds $t_1=t_2=t$, as
well as for the variants with data-dependent thresholds. There is a
clear gain from adaptivity since the power curves for the adaptive
procedures  are above those for the non-adaptive procedures, for the
same fixed threshold $t$. The choice
of $t$ matters, and the choice $t=0.025$ is better than the choice
$t=0.05$, and fairly close to the best $t$.  We see that the power of the
non-adaptive procedures with data-dependent selection thresholds is
superior to the power of non-adaptive procedures with fixed thresholds. The same
is true for the adaptive procedures in all the settings except for
the last two rows of the equi-correlation setting, where the power
of the adaptive procedures with data-dependent thresholds is
slightly lower than the highest power for fixed thresholds
$t_1=t_2=t.$ In these settings the number of selected hypotheses is
on average lower than in other settings, and the fractions of true
null hypotheses in one study among the selected in the other study
are expected to be small. As a result, the solutions to the two non-linear equations solved using the estimates of the fractions of nulls  are far from optimal. 
Therefore, when there is dependence within each study, and the number of selected hypotheses is small (say less than 100 per study), we suggest using the novel adaptive procedures with $t_1=t_2=\alpha/2$ instead of using  data-dependent $(t_1,t_2)$.

Figure \ref{fig-fdrproc} shows the power and FDR versus $\mu$ under
independence (columns 1 and 2) and under equi-correlation of the
test statistics with $\rho=0.25$ (columns 3 and 4).  The novel procedures are clearly superior to the
competitors: the empirical Bayes procedure does not control the FDR
when $m=1000$, and the actual level reaches above 0.1 under
dependence; the oracleMax and adaptiveMax procedures have the lowest power
in almost all settings. The novel adaptive procedures approach the power
of the oracle Bayes as $f_{10}=f_{01}$ increase. The adaptive
procedures with $\lambda=0.05$ and  $\lambda=0.5$ have similar
power, but the FDR with $\lambda=0.05$ is controlled in all
dependence settings and the FDR with $\lambda=0.5$ is above the
nominal level in three of the dependence settings. Our results concur
with the results of \cite{Blanchard09} for single studies, that the
preferred parameter is $\lambda=0.05$. We thus recommend the
adaptive FDR-replicability procedure with $\lambda =\alpha$, for FDR control at level $\alpha$. We also recommend using  
data-dependent $(t_1,t_2)$, unless the test statistics are dependent within each study and the number of selected hypotheses from each study is expected to be small.

\begin{figure}[htbp]
    \centering
    \includegraphics[width =17.5cm,height = 22.5cm]{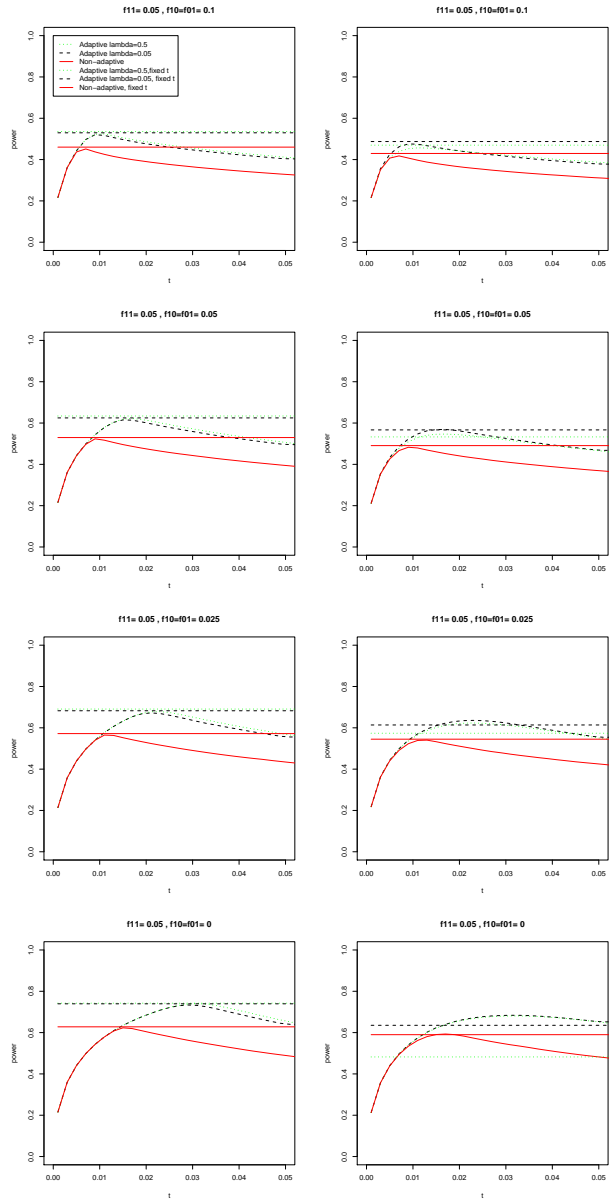}
\caption{Column 1: Independence setting; Columns 2: Equi-correlation
with $\rho=0.25$. The power versus fixed threshold $t$ is shown for
$\mu = 3$ for the adaptive and non-adaptive FDR-replicability
procedures, along with the power of these procedures with
data-dependent thresholds. In all settings $m=1000$, $\alpha=0.05,
\alpha_1=0.025$. }\label{fig:fdrFixedThresholds}
\end{figure}

\begin{figure}[htbp]
  %  \centering
  %  \begin{tabular}{c}
    \includegraphics[width =17.5cm,height = 22.5cm]{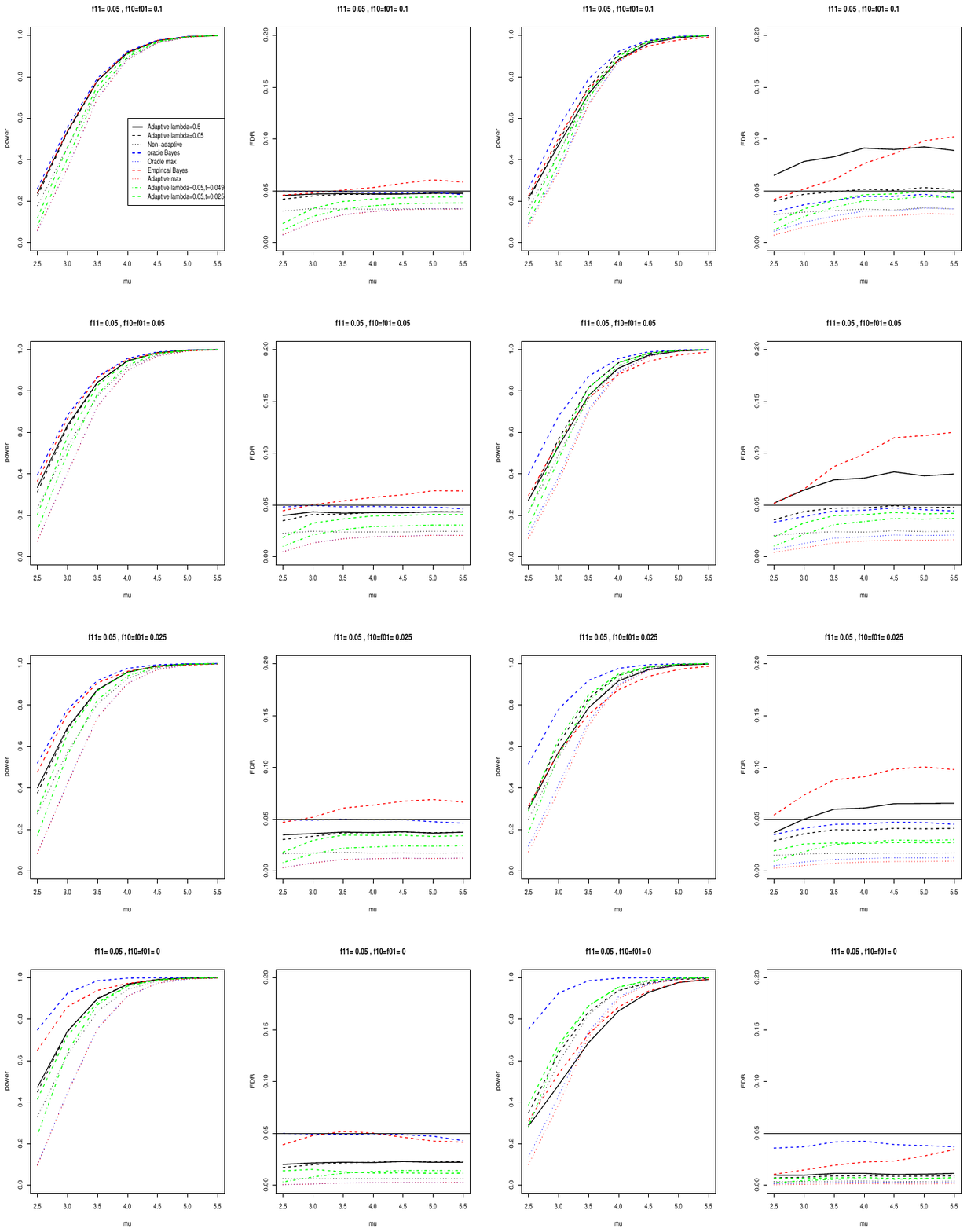}
   % \end{tabular}
        \caption{\scriptsize{Columns 1 and 2: Independence setting;
    Columns 3 and 4: Equi-correlation with $\rho=0.25$.
    Power and FDR versus $\mu$ for the adaptive-FDR-replicability procedure with data-dependent $(t_1,t_2)$ with $\lambda=0.5$ (solid black) and with $\lambda=0.05$ (dashed black); Non-adaptive-FDR-replicability procedure with data-dependent $(t_1,t_2)$ (dotted black);
    the oracle Bayes (dashed blue) and empirical Bayes (dashed red); the oracle and adaptive BH on maximum $p$-value,  (dotted blue and dotted red);
    adaptive-FDR-replicability procedure with $\lambda=0.05$ and fixed $t_1=t_2=0.049$ (dash-dot green) and fixed $t_1=t_2=0.025$ (dash green).
    In all settings $m=1000$, $\alpha=0.05, \alpha_1=0.025$.
    }}
    \label{fig-fdrproc}
\end{figure}

\section{Examples}\label{sec-example}
\subsection{Laboratory mice studies comparing behaviour across strains}
%%%%PROBLEM WITH TABLE LABELS%%%%%%%%%%%%%%%%%%%%%%%%%%%%%%%%%%%%%%%%%%%%%%%%%
It is well documented that in different laboratories, the comparison
of behaviors of the same two strains may lead to opposite
conclusions that are both statistically significant
(\cite{Crabbe99}, \cite{kafkafi05}, and Chapter 4 in
\cite{Crusio13}). An explanation may be the different laboratory
environment (i.e. personnel, equipment, measurement techniques)
affecting differently the study strains (i.e. an interaction of
strain with laboratory).  \cite{richter11} examined 29 behavioral
measures from five commonly used behavioral tests (the barrier test,
the vertical pole test, the elevated zero maze, the open field test,
and the novel object test) on female mice from different strains in
different laboratories with standardized conditions. Table 1
 shows the one-sided $p$-value in the direction
favored by the data based on the comparison of two strains in two
laboratories, for each of the 29 outcomes.

\begin{center}
\begin{table}
\caption{\scriptsize{For 16 female mice from each of two inbred
strains, " C57BL6NCrl" and  "DBA/2NCrl", in each of two
laboratories, the Wilcoxon rank sum test one-sided $p$-value was
computed for the test of no association between strain and
behavioral endpoint. We show the $p$-values for the lab of H. Wurbel
at the University of Giessen in column 3, and for the lab of P. Gass
at the  Central Institute of Mental Health, Mannheim  in column 4.
The direction of the alternative favored by the data is shown in
column 2, and it is marked as "X" if the laboratories differ in the
direction of smallest one-sided $p$-value. The rows are the outcomes from 5
behavioural tests: the barrier test (row 1); the vertical pole test
(row 2); the elevated zero maze (rows 3-11) ; the open field test
(rows 12-19); the novel object test (rows
20-29).}}\label{tab-mice-pv} \centering
  \begin{tabular}{cc}
  \scriptsize

\centering
\begin{tabular}{cccc}
  \hline
  & & \multicolumn{2}{c}{$\min(P_{ij}^L, P_{ij}^R)$}\\
 & Alternative & $i=1$ &  $i=2$ \\
  \hline
1 & X & 0.3161 & 0.0218 \\
  2 & $C57BL<DBA$ & 0.0012 & 0.0000 \\
  3 & X & 0.0194 & 0.1120 \\
  4 & $C57BL<DBA$ & 0.0095 & 0.2948 \\
  5 & $C57BL<DBA$ & 0.1326 & 0.0028 \\
  6 & $C57BL>DBA$ & 0.1488 & 0.0003 \\
  7 & $C57BL>DBA$ & 0.2248 & 0.0000 \\
  8 & X & 0.4519 & 0.0005 \\
  9 & $C57BL<DBA$ & 0.0061 & 0.0000 \\
  10 & $C57BL<DBA$ & 0.0071 & 0.0888 \\
  11 & X & 0.4297 & 0.1602 \\
  12 & $C57BL<DBA$ & 0.0918 & 0.0506 \\
  13 & X & 0.0918 & 0.0001 \\
  14 & $C57BL<DBA$ & 0.0000 & 0.0048 \\
  15 & X & 0.0005 & 0.0550 \\
   \hline
\end{tabular}
%\end{table}
&
%\begin{table}[ht]
%\centering
\scriptsize
\begin{tabular}{cccc}
  \hline
 & & \multicolumn{2}{c}{$\min(P_{ij}^L, P_{ij}^R)$}\\
 & Alternative & $i=1$ &  $i=2$ \\
  \hline

  16 & $C57BL<DBA$ & 0.0059 & 0.0002 \\
  17 & $C57BL>DBA$ & 0.0176 & 0.0003 \\
  18 & X & 0.0000 & 0.0538 \\
  19 & $C57BL<DBA$ & 0.0000 & 0.1727 \\
  20 & $C57BL<DBA$ & 0.0157 & 0.0001 \\
  21 & $C57BL<DBA$ & 0.0000 & 0.0234 \\
  22 & $C57BL<DBA$ & 0.3620 & 0.0176 \\
  23 & $C57BL<DBA$ & 0.0000 & 0.0001 \\
  24 & $C57BL<DBA$ & 0.0000 & 0.0076 \\
  25 & $C57BL<DBA$ & 0.0000 & 0.0000 \\
  26 & $C57BL<DBA$ & 0.0000 & 0.0003 \\
  27 & $C57BL<DBA$ & 0.0000 & 0.0001 \\
  28 & $C57BL<DBA$ & 0.0000 & 0.0550 \\
  29 & X & 0.0033 & 0.3760 \\
  & & & \\
   \hline
\end{tabular}
%\end{table}
\end{tabular}
\end{table}
\end{center}

The example is too small for considering the empirical Bayes
approach. The approach suggested in \cite{benjamini09} of using for
each feature the  maximum of the two studies $p$-values, i.e., $2\min \{\max(p_{1j}^L, p_{2j}^L), \max(p_{1j}^R,
p_{2j}^R) \}$, detected overall fewer outcomes than using our novel
procedures both for FWER and for FDR control.

Table 2 shows the FWER/FDR non-adaptive and adaptive $r$-values, for the selected features,
according to the rule which selects all features with two-sided
$p$-values that are at most 0.05. We did not consider data-dependent
thresholds since the number of features examined was only 29, which
could result in highly variable data-dependent thresholds and a
power loss comparing to procedures with fixed thresholds, as was
observed in simulations. At the $\alpha=0.05$ level, for FWER
control, four discoveries were made by using Bonferroni on the
maximum $p$-values, and five discoveries were made with the
non-adaptive and adaptive Bonferroni-replicability  procedures. At
the $\alpha=0.05$ level, for FDR control, nine discoveries were made
by using BH on the maximum $p$-values, and nine and twelve
discoveries were made with the non-adaptive FDR and adaptive
FDR-replicability procedures, respectively.  Note that the adaptive
$r$-values can be less than half the non-adaptive $r$-values, since
$\prone = 0.44$ and $\prtwo =0.47$.

\begin{center}
\begin{table}
%\resizebox{1.4\textwidth}{!}{\begin{minipage}{\textwidth}
\caption{\scriptsize{The replicability analysis results for the data
in Table 1, after selection of features with
two-sided $p$-values at most 0.05 (i.e. $t_1=t_2=0.025$). Only the
twelve features in $\mathcal S_1\cap \mathcal S_2$ are shown, where
$S_1=20$, $S_2=19$. For each selected feature, we show the
$r$-values based on Bonferroni (column 5),  FDR (column 6), adaptive
Bonferroni (column 7), and the adaptive FDR (column 8). The adaptive
procedures used $\lambda =0.05$.}} \label{tab:miceresults}
 \centering
 \scriptsize
\begin{tabular}{rrrrrrrrr}
  \hline
index & & \multicolumn{2}{c}{$\min(P_{ij}^L, P_{ij}^R)$}&\multicolumn{2}{c}{Non-adaptive}&\multicolumn{2}{c}{Adaptive}\\
selected & Alternative & $i=1$ &  $i=2$  & Bonf & FDR & Bonf & FDR \\
  \hline
 2 & $C57BL<DBA$ & 0.0012 & 0.0000 & 0.0452 & 0.0090 & 0.0200 & 0.0040 \\
   9 & $C57BL<DBA$ & 0.0061 & 0.0000 & 0.2323 & 0.0290 & 0.1029 & 0.0129 \\
   14 & $C57BL<DBA$ & 0.0000 & 0.0048 & 0.1910 & 0.0290 & 0.0905 & 0.0129 \\
   16 & $C57BL<DBA$ & 0.0059 & 0.0002 & 0.2237 & 0.0290 & 0.0992 & 0.0129 \\
   17 & $C57BL>DBA$ & 0.0176 & 0.0003 & 0.6679 & 0.0607 & 0.2960 & 0.0269 \\
   20 & $C57BL<DBA$ & 0.0157 & 0.0001 & 0.5974 & 0.0597 & 0.2648 & 0.0265 \\
   21 & $C57BL<DBA$ & 0.0000 & 0.0234 & 0.9363 & 0.0780 & 0.4435 & 0.0370 \\
   23 & $C57BL<DBA$& 0.0000 & 0.0001 & 0.0022 & 0.0011 & 0.0010 & 0.0005 \\
   24 & $C57BL<DBA$ & 0.0000 & 0.0076 & 0.3037 & 0.0337 & 0.1439 & 0.0160 \\
   25 & $C57BL<DBA$ & 0.0000 & 0.0000 & 0.0005 & 0.0005 & 0.0003 & 0.0003 \\
   26 & $C57BL<DBA$ & 0.0000 & 0.0003 & 0.0126 & 0.0032 & 0.0060 & 0.0015 \\
   27 & $C57BL<DBA$ & 0.0000 & 0.0001 & 0.0038 & 0.0013 & 0.0018 & 0.0006 \\
   \hline
%\label{tab:miceresults}
\end{tabular}
\end{table}
\end{center}

\subsection{Microarray studies comparing groups with different cancer
severity}\label{subsec-big-example} \cite{Freije04} and
\cite{Phillips06} compared independently the expression levels in
patients with grade $III$ and grade $IV$ brain cancer. Both studies
used the Affymetrix HG U133 oligonucleotide arrays, with 22283
probes in each study.
The study of \cite{Freije04} (GEO  accession GSE4412) included  26 subjects with tumors diagnosed as grade III glioma and 59 subjects with tumor diagnosis of grade IV glioma, all undergoing surgical treatment at the university of California, Los Angeles. %Patients with grade III tumors had a median survival of 4.8 years, whereas patients with grade IV tumors had a median survival of 293 days.
The study of \cite{Phillips06} (GEO accession GSE4271) included 24 grade III subjects, and 76 grade IV subjects, from the M.D. Anderson Cancer Center (MDA). %Tumor grade is the most established and robust predictor of disease
 The Wilcoxon rank sum test $p$-values were computed for each probe in each study in order to quantify the evidence against no association of probe measurement with tumor subgroup. %We will demonstrate how we can identify which probe sets correlate with the tumor subgroup in both studies, while controlling for false positives.

We used the R package \emph{repfdr} \citep{Heller14c} to get the following estimated fractions, among the 22283 probes:
0.39  with $\vec h = (0,0)$; 0.16 with $\vec h = (1,1)$; 0.13 with $\vec h = (-1,-1)$; 0.10 with $\vec h = (0,1)$; 0.08 with $\vec h = (-1,0)$; 0.07 with $\vec h = (0,-1)$; 0.07 with $\vec h = (1,0)$; 0.00 with $\vec h = (-1,1)$ or $\vec h = (1,-1)$.

For FWER-replicability, the recommended Procedure
\ref{proc-bonferroniadapt} with $\lambda=0.05$ and data-dependent
thresholds $t_1 = 6.5*10^{-5}, t_2 = 5.1*10^{-5}$ discovered 340
probes. For comparison, the non-adaptive and adaptive
Bonferroni-replicability procedure with fixed thresholds $t_1
=t_2=0.025$ discovered only 90 and 124 probes, respectively. The
Bonferroni on maximum $p$-values discovered only 47 probes.

For FDR-replicability, the recommended adaptive procedure in Section
\ref{sec-adapt} with $\lambda=0.05$ and data-dependent thresholds
$t_1 = 0.021, t_2 = 0.024$ discovered 3383 probes. For comparison,
the non-adaptive and adaptive FDR-replicability procedure with fixed
selection thresholds $t_1 =t_2=0.025$  discovered 2288 and 3299
probes, respectively. The adaptive $r$-values can be half the non-adaptive $r$-values, since 
$\hat \pi_0^I =0.51$ and $\hat \pi_0^{II} =0.49$.
Among the two competing approaches, the BH on
maximum $p$-values discovered only 1238 probes, and  the empirical
Bayes procedure discovered 4320 probes.  Among the 3383 probes
discovered by our approach, 3377 were also discovered by the
empirical Bayes procedure.

\section{Discussion}\label{sec-Discussion}
In this paper we proposed novel procedures for establishing
replicability in two studies. First, we introduced procedures that
take the selected set of features in each of two studies, and infer
about the replicability of features selected in both studies while
controlling for false replicability claims. We proved that the FWER
controlling procedure is valid (i.e., controls the error rate at the
desired nominal level) for any dependence within each study, and
that the FDR controlling procedure is valid under independence of
the test statistics within each study, and suggested also a more
conservative procedure that is valid for arbitrary dependence. %The sentence is long, maybe better two sentences?
Next, we suggested incorporating the plug-in estimates of the
fraction of nulls in one study among the selected  features by the
other study, which can be estimated as long as the $p$-values for
the union of features selected is available. We proved that the
resulting adaptive FWER and FDR controlling procedures are valid
under independence of the test statistics within each study. Our
empirical investigations showed that the adaptive procedures remain
valid even when the independence assumption is violated, as long as
we use $\lambda = \alpha$ as a parameter for the plug-in estimates,
as suggested by \cite{Blanchard09} for the adaptive BH procedure.
Finally, when two full studies are available that examine the same
features, we suggested selecting features for replicability analysis
that have $p$-values below certain thresholds. We showed that
selecting the features with one-sided $p$-values below $\alpha/2$ has good
power, but that the power can further be improved if we use
data-dependent thresholds, which receive the values that will lead
to the procedure selecting exactly the features that are discovered
as having replicated findings.

Our practical guidelines for establishing replicability are to use
the adaptive procedure for the desired error rate control, with
$\lambda = \alpha$. Moreover, based on the simulation results we
suggest using the data-dependent selection thresholds when two
full studies are available if the number of selected features in each study is
expected to be large enough (say above 100), and using the fixed thresholds $t_1=t_2=\alpha/2$ otherwise.
We would like to note that the $r$-value computation is more
involved when the thresholds are data-dependent, since these
thresholds depend on the nominal level $\alpha$.  An interesting
open question is how to account for multiple solutions of the two
non-linear equations that are solved in order to find the
data-dependent thresholds.

The suggested procedures can be generalized to the case that  more
than two studies are available. It is possible to either aggregate
multiple results of pairwise replicability analyses, or to first
aggregate the data and then apply a single replicability analysis on
two meta-analysis $p$-values. The aim of the replicability analysis
may also be redefined to be that of discovering features that have
replicated findings in at least $u$ studies, where $u$ can range
from two to the total number of studies. Other extensions include
weighting the features differently, as suggested by
\cite{Genovese06}, based on prior knowledge on the features, and
replicability analysis on multiple families of hypotheses while
controlling more general error rates, as suggested by
\cite{Benjamini13}.

\appendix

\section{Notation for technical derivations}\label{app-notation}

For the technical derivations, the following notation will be used.
Let $p_i$ be the $m$-dimensional vector of $p$-values for study $i,$
$\mathcal{S}_i(p_i)$ be the index set of features selected from
study $i$ based on the vector of $p$-values $p_i,$ and $S_i(p_i)$ be
the cardinality of this set, for $i\in\{1,2\}.$ Let $P_i^{(j)} =
(P_{i1}, \ldots, P_{i,j-1}, P_{i,j+1}, \ldots, P_{im})$ be the
vector of $p$-values for the $m-1$ features excluding $j$,  for
$i=1,2$. When the selection rule by which the set $\mathcal{S}_i$ is
selected is stable, define $\mathcal S_{i}^{(j)}\subseteq\{1,\ldots,
j-1, j+1,\ldots,m\}$ as the set of indices selected along with $j,$
if $j\in \mathcal{S}_i,$ and $\mathcal S_{i,\lambda}^{(j)}$ as
$\mathcal{S}_i^{(j)}\cap \{l\neq j: P_{il}\leq \lambda\}$ if $j\in
\mathcal{S}_{i, \lambda},$ for $i\in\{1,2\},$ and let $S_{i}^{(j)} =
|\mathcal S_{i}^{(j)}|$. Define $\mathcal
S_{i}^{(j)}(t_i)=\{l:p_{il}\leq t_i, l\neq j\}$ as the index set of
features with $p$-value at most $t_i$ from the vector of $p$-values
$p_i^{(j)}$, and let $S_{i}^{(j)}(t_i) =
|\mathcal S_{i}^{(j)}(t_i)|$. For $c\in(0,1),$ we write $\alpha_1=c\alpha$ and $\alpha_2 = \alpha-\alpha_1.$ %and $q_2 = q-q_1$.

\section{Proof of Theorems \ref{indep} and \ref{thm-proc-fdr-adaptive}}\label{app-thm-fdr-indep}
In the proofs of Theorems \ref{indep} and
\ref{thm-proc-fdr-adaptive} we use the following lemma. The lemma is
proven in the end of the section.
\begin{lemma}\label{lem_fdr}
Let $\mathcal{S}_i $ be the selected set of features based on study
$i$, for $i=1,2$. Let $r_j,\,\, j\in \mathcal{S}_1\cap
\mathcal{S}_2$ be the Bonferroni-type $r$-values:
\begin{align}
r_j=\max\left\{\frac{W_1p_{1j}}{c}, \frac{W_2p_{2j}}{1-c}\right\},
\,\,j\in \mathcal{S}_1\cap \mathcal{S}_2,
\end{align}
where $c\in(0,1)$ is a constant and $W_1, W_2$ may be constants or
random variables based on $p$-values. The FDR $r$-values
based on the Bonferroni-type $r$-values are:
$$ r^{FDR}_j = \min_{\{i:\, r_i\geq r_j, i \in \mathcal{S}_1\cap \mathcal{S}_2  \}} \frac{r_i}{rank(r_i)},\,\,j\in \mathcal{S}_1\cap \mathcal{S}_2,$$
where $rank(r_i)$ is the rank of the Bonferroni-type $r$-value for
feature $i \in \mathcal{S}_1\cap \mathcal{S}_2$, with maximum rank
for ties.

\begin{enumerate}
\item[(1)] The procedure that declares as replicated
the features with FDR r-values at most $\alpha$ is equivalent to the
following procedure on the selected features $\mathcal{S}_1\cap
\mathcal{S}_2$:
\begin{enumerate}
\item Let $$R\triangleq\max\left\{r:
\sum_{j\in\mathcal{S}_1\cap \mathcal{S}_2}\textbf{I}\left[(p_{1j},
p_{2j})\leq\left(\frac{r\alpha_1}{W_1},
\frac{r\alpha_2}{W_2}\right)\right] = r\right\}.$$ %where
%$\alpha_1=c\alpha.$
\item The set of indices with replicability claims is
\begin{align}\mathcal R= \{j: (p_{1j},
p_{2j})\leq\left(\frac{R\alpha_1}{W_1},
\frac{R\alpha_2}{W_2}\right), \,\,j \in \mathcal S_1\cap \mathcal
S_2\}.\notag\end{align}
\end{enumerate}
\item [(2)] The procedure that declares as replicated the features with FDR r-values at most $\alpha$ controls the FDR for
replicability analysis at level $\alpha$ if the following conditions
are satisfied:
\begin{enumerate}
\item The $p$-values corresponding to true null hypotheses are each independent of all the other $p$-values.
\item For each $j\in\{1,\ldots,m\},$ there exist random variables (or constants) $W_{1}^{(j)}, W_{2}^{(j)}$  defined on
the space $(P_1^{(j)}, P_2^{(j)})$ such that if $j\in
\mathcal{S}_1\cap \mathcal{S}_2,$ then $W_1=W_{1}^{(j)},$ $W_2=W_{2}^{(j)},$%for all $i\in
%\mathcal{S}_1^{(j)}\cap \mathcal{S}_2^{(j)},$ $r_i=r_i^{(j)},$ where
%$$r_i^{(j)}\triangleq\max\left\{\frac{W_1^{(j)}p_{1j}}{c},
%\frac{W_2^{(j)}p_{2j}}{1-c}\right\}, \,\,i\in
%\mathcal{S}_1^{(j)}\cap \mathcal{S}_2^{(j)},$$
and for arbitrary fixed vectors $p_1$%=(p_{11},\ldots, p_{1m}),$
and $p_2$ it holds:%=(p_{21}, \ldots, p_{2m}),$
\begin{align}
\textbf{I}[j\in
\mathcal{S}_2(p_2)]E(1/W_{1}^{(j)}\,|\,P_2=p_2)\leq\frac{1}{\sum_{j\in
\mathcal{S}_2(p_2)}(1-H_{1j})}\label{wj1}
\end{align}
\begin{align}
\textbf{I}[j\in
\mathcal{S}_1(p_1)]E(1/W_{2}^{(j)}\,|\,P_1=p_1)\leq\frac{1}{\sum_{j\in
\mathcal{S}_1(p_1)}(1-H_{2j})}.\label{wj2}
\end{align}
%where $\mathcal{S}_1(p_1)$ and $\mathcal{S}_2(p_2)$ are the selected
%features from study one and study two based on $P_1=p_1$ and
%$P_2=p_2,$ respectively.
\end{enumerate}
\end{enumerate}
\end{lemma}
\textbf{Proof of item 1 of Theorem \ref{indep} } The result of item
1 of Theorem \ref{indep} follows from Lemma \ref{lem_fdr}. The
conditions of Lemma \ref{lem_fdr} hold with $W_1=S_2,$
$W_{1}^{(j)}=1+S_{2}^{(j)},$ and $W_2=S_1,$
$W_2^{(j)}=1+S_{1}^{(j)}.$ In order to see it, note that for $j\in
\mathcal{S}_1\cap \mathcal{S}_2,$ $S_i=1+S_{i}^{(j)}$ for $i\in
\{1,2\}.$ In addition, note that for arbitrary fixed vector $p_2$
and $j\in\{1,\ldots,m\}$
\begin{align*}\textbf{I}[j\in
\mathcal{S}_2(p_2)]E\left(\frac{1}{W_1^{(j)}}\,|\,P_2=p_2\right)=&\textbf{I}[j\in
\mathcal{S}_2(p_2)]E\left(\frac{1}{1+S_2^{(j)}}\,|P_2=p_2\right)\\=&\textbf{I}[j\in
\mathcal{S}_2(p_2)]\left(\frac{1}{S_2(p_2)}\,\right)\leq\frac{1}{\sum_{j\in
\mathcal{S}_2(p_2)}(1-H_{1j})}.\notag\end{align*} Thus we have
proved inequality (\ref{wj1}). The proof of inequality (\ref{wj2})
is similar. 

\noindent \textbf{Proof of Theorem
\ref{thm-proc-fdr-adaptive}} The result of Theorem
\ref{thm-proc-fdr-adaptive} follows from Lemma \ref{lem_fdr}. The
conditions of Lemma \ref{lem_fdr} hold with
$\mathcal{S}_{i,\lambda}, i=1,2$ as the selected sets,
%$\mathcal{S}_1=\mathcal{S}_{1,\lambda},$
%$\mathcal{S}_2=\mathcal{S}_{2,\lambda}$
and
$$W_1=S_{2,\lambda}\prone=\frac{1+\sum_{i\in
\mathcal{S}_{2,\lambda}}\textbf{I}(P_{1i}>\lambda)}{(1-\lambda)},\,\,W_1^{(j)}=\frac{1+\sum_{i\in
\mathcal{S}_{2,\lambda, i\neq
j}}\textbf{I}(P_{1i}>\lambda)}{(1-\lambda)},$$
$$W_2=S_{1,\lambda}\prtwo=\frac{1+\sum_{i\in
\mathcal{S}_{1,\lambda}}\textbf{I}(P_{2i}>\lambda)}{(1-\lambda)},\,\,W_2^{(j)}=\frac{1+\sum_{i\in
\mathcal{S}_{1,\lambda, i\neq
j}}\textbf{I}(P_{2i}>\lambda)}{(1-\lambda)}.$$ In order to see it,
note that if $j\in \mathcal{S}_{1,\lambda}\cap
\mathcal{S}_{2,\lambda},$ it holds that $\max\{P_{1j}, P_{2j}\}\leq
\lambda.$ In addition, it was shown in the proof of Theorem
\ref{thm-bonferroniadapt}  that for arbitrary fixed vector
$p_2$ and $j\in\{1,\ldots,m\}$%=(p_{21}, \ldots, p_{2m})$
\begin{align*}&\textbf{I}[j\in
\mathcal{S}_{2,\lambda}
(p_2)]E\left(\frac{1}{W_1^{(j)}}\,|\,P_2=p_2\right)\leq\\&\textbf{I}[j\in
\mathcal{S}_{2,\lambda} (p_2)]E\left(1/\left(\frac{1+\sum_{i\in
\mathcal{S}_{2,\lambda, i\neq
j}}(1-H_{1j})\textbf{I}(P_{1i}>\lambda)}{1-\lambda}\right)\,|\,P_2=p_2\right)\leq\\&\frac{1}{\sum_{j\in
\mathcal{S}_{2,\lambda}(p_2)}(1-H_{1j})}.\notag\end{align*}
The proof of inequality (\ref{wj2}) is similar.  \\
\noindent
%It can be similarly shown that for arbitrary fixed vector $p_1$ and $j\in\{1,\ldots,m\},$
%\begin{align*}&\textbf{I}[j\in
%\mathcal{S}_{1,\lambda}
%(p_1)]E\left(\frac{1}{W_2^{(j)}}\,|\,P_1=p_1\right)\leq\frac{1}{\sum_{j\in
%\mathcal{S}_{1,\lambda}(p_1)}(1-H_{2j})}.\notag\end{align*}
\textbf{Proof of item 1 of Lemma \ref{lem_fdr}.}
Note that the procedure given in item 1 of Lemma \ref{lem_fdr} can
be written as follows:
\begin{enumerate}
\item Let $$R=\max\left\{r:
\sum_{j\in\mathcal{S}_1\cap \mathcal{S}_2}\textbf{I}\left[r_j\leq
r\alpha\right] = r\right\}.$$
\item The set of indices with replicability claims is
\begin{align}\mathcal R= \{j: r_j\leq R\alpha, \,\,j \in \mathcal S_1\cap
\mathcal S_2\}.\notag\end{align}
\end{enumerate}
Let $r_0 = \max\left\{r: \sum_{j\in\mathcal{S}_1\cap
\mathcal{S}_2}\textbf{I}\left[r_j\leq r\alpha\right] \geq
r\right\}.$ We prove that $R=r_0$ by contradiction.
%Assume that there exists $r_0>R$ such that $\sum_{j\in\mathcal \mathcal{S}_1\cap
%\mathcal{S}_2}\textbf{I}\left[r_j\leq r_0\alpha\right] > r_0. $
%Then it follows from the definitions of $r_0$ and $R$ that
From the definitions of $r_0$ and $R$ it follows that if $R\neq
r_0,$ then $r_0>R,$ and $\sum_{j\in\mathcal{S}_1\cap
\mathcal{S}_2}\textbf{I}\left[r_j\leq r_0\alpha\right] \geq r_0+1. $
However, since $\sum_{j\in\mathcal{S}_1\cap
\mathcal{S}_2}\textbf{I}\left[r_j\leq (r_0+1)\alpha\right]\geq
\sum_{j\in\mathcal{S}_1\cap \mathcal{S}_2}\textbf{I}\left[r_j\leq
r_0\alpha\right] $ it follows that $r_0+1$ is also in $\left\{r:
\sum_{j\in\mathcal{S}_1\cap \mathcal{S}_2}\textbf{I}\left[r_j\leq
r\alpha\right] \geq r\right\}$, thus contradicting the definition of
$r_0$ as being the greatest value in this set.
%Moreover, let us show that \begin{align}R=\max\left\{r:
%\sum_{j\in\mathcal \mathcal{S}_1\cap
%\mathcal{S}_2}\textbf{I}\left[r_j\leq r\alpha\right] \geq
%r\right\}.\label{def_r}\end{align} By contradiction, assume that
%there exists $r_0>R$ such that $\sum_{j\in\mathcal \mathcal{S}_1\cap
%\mathcal{S}_2}\textbf{I}\left[r_j\leq r_0\alpha\right] > r_0. $ Then
%$$\sum_{j\in\mathcal \mathcal{S}_1\cap
%\mathcal{S}_2}\textbf{I}\left[r_j\leq (r_0+1)\alpha\right]\geq
%\sum_{j\in\mathcal \mathcal{S}_1\cap
%\mathcal{S}_2}\textbf{I}\left[r_j\leq r_0\alpha\right] \geq r_0+1.
%$$ Since $r_0+1>R,$ it follows from the above inequality and the
%definition of $R$ that $$\sum_{j\in\mathcal \mathcal{S}_1\cap
%\mathcal{S}_2}\textbf{I}\left[r_j\leq (r_0+1)\alpha\right]>r_0+1.$$
%Applying this argument repeatedly $|\mathcal{S}_1\cap
%\mathcal{S}_2|-r_0$ times, we obtain
%$$\sum_{j\in\mathcal \mathcal{S}_1\cap
%\mathcal{S}_2}\textbf{I}\left[r_j\leq |\mathcal{S}_1\cap
%\mathcal{S}_2|\alpha\right]>|\mathcal{S}_1\cap \mathcal{S}_2|.$$
%This is a contradiction, since obviously $\sum_{j\in\mathcal
%\mathcal{S}_1\cap \mathcal{S}_2}\textbf{I}\left[r_j\leq
%|\mathcal{S}_1\cap \mathcal{S}_2|\alpha\right]\leq|\mathcal{S}_1\cap
%\mathcal{S}_2|.$
Thus we have proved that \begin{align}R=\max\left\{r:
\sum_{j\in\mathcal{S}_1\cap \mathcal{S}_2}\textbf{I}\left[r_j\leq
r\alpha\right] \geq r\right\}.\label{def_r}\end{align}

%XXXX CONTINUE FROM HERE

We now prove that the procedure that declares as replicated the
features with FDR r-values at most $\alpha$ is equivalent to the
following procedure in item 1, i.e.
\begin{align}\{j: r_j^{FDR}\leq \alpha\}=\{j: r_j\leq R\alpha \},\label{eqsets}\end{align} where
$R$ is given in (\ref{def_r}). Let us first prove that
\begin{align}\{j: r_j^{FDR}\leq \alpha\}\subseteq\{j: r_j\leq R\alpha
\}.\label{muhal1}\end{align} Let $j\in\{j: r_j^{FDR}\leq \alpha\}$
be arbitrary fixed. There exists $i_0\in \mathcal{S}_1\cap
\mathcal{S}_2$ such that $r_{i_0}\geq r_j$ and
$$\frac{r_{i_0}}{rank(r_{i_0})}=\min_{\{i:\, r_i\geq r_j, i\in
\mathcal{S}_1\cap \mathcal{S}_2\}}\frac{r_i}{rank(r_i)}\leq
\alpha.$$ Thus $r_{i_0}\leq rank(r_{i_0})\alpha.$ Therefore,
$rank(r_{i_0})\leq \sum_{j\in\mathcal{S}_1\cap
\mathcal{S}_2}\textbf{I}\left[r_j\leq rank(r_{i_0})\alpha\right]$.
This inequality and the expression for $R$ given in (\ref{def_r})
yield that $rank(r_{i_0})\leq R.$ It follows that $r_{i_0}\leq
R\alpha.$ Recall that $r_j\leq r_{i_0},$ therefore $r_j\leq
R\alpha.$ Thus we have proved $(\ref{muhal1}).$ Let us now prove
that
\begin{align}\{j: r_j\leq R\alpha
\}\subseteq\{j: r_j^{FDR}\leq \alpha\}.\label{muhal2}\end{align} Let
$j\in \mathcal{S}_1\cap \mathcal{S}_2$ be an arbitrary fixed index
such that $r_j\leq R\alpha.$ Since $r_j\leq r_{(R)}$, and
$\frac{r_{(R)}}{R}\leq \alpha$ (where $r_{(R)}$ is the $R$'th
largest $r$-value), it follows that
% Define $$i_0=\arg\max_{l\in
%\mathcal{S}_1\cap \mathcal{S}_2}\{r_l:r_l\leq R\alpha\}.$$
% From the definition of $i_0$ it
%follows that $r_j\leq r_{i_0},$ and
%$$rank(r_{i_0})=\sum_{j\in\mathcal \mathcal{S}_1\cap
%\mathcal{S}_2}\textbf{I}\left[r_j\leq R\alpha\right]=R.$$ Since
%$r_{i_0}\leq R\alpha,$ we obtain from the above that
%$\frac{r_{i_0}}{rank(r_{i_0})}\leq \alpha.$ Recall that $i_0\in
%\mathcal{S}_1\cap \mathcal{S}_2$ and $r_j\leq r_{i_0},$ therefore
$$r_j^{FDR}=\min_{\{i:\, r_i\geq r_j, i\in S_1\cap
S_2\}}\frac{r_i}{rank(r_i)}\leq \alpha.$$ Thus we have proved
(\ref{muhal2}), which completes the proof of (\ref{eqsets}) and of
item 1.
\\\textbf{Proof of item 2 of Lemma \ref{lem_fdr}}
For $j\in\{1,\ldots,m\}$ let us define $C_k^{(j)}$ as the event in
which if $r_j^{FDR}\leq \alpha,$ then the total number of FDR
$r$-values which are at most $\alpha$ is $k.$ It follows from item 1
and from condition (\textit{ii}) of item 2 that the event
$C_k^{(j)}$ is defined on the space $(P_1^{(j)}, P_2^{(j)})$ as
follows. Let
\begin{align}T_{i}^{(j)}=\left\{
\begin{array}{cl}
\max\left(\frac{W_1^{(j)}p_{1i}}{c}, \frac{W_2^{(j)}p_{2i}}{1-c} \right) & \text{if}\,\, i\in \mathcal{S}_1^{(j)}\cap \mathcal{S}_2^{(j)},\\
 \infty& \text{otherwise. } \\
\end{array} \right. \notag\end{align}
and let $T^{(j)}_{1}\leq \ldots\leq T^{(j)}_{m-1}$ be the sorted
$T$-values, where we set $T_0^{(j)}=0.$ Note that $T_{i}^{(j)}=r_i$
for $i\in \mathcal{S}_1^{(j)}\cap \mathcal{S}_2^{(j)}.$ It follows
from the equivalent procedure given in item 1 of Lemma \ref{lem_fdr}
that
\begin{align}C_k^{(j)}=\{(P_1^{(j)},
P_2^{(j)}):\,T^{(j)}_{(k-1)}\leq k\alpha, T^{(j)}_{(k)}>
(k+1)\alpha,\ldots,T^{(j)}_{(m-1)}> m\alpha
\}.\label{ckj-lemma}\end{align} Note that given $P_1=p_1,$ for $j\in
\mathcal{S}_1(p_1),$ $C_k^{(j)}=\emptyset$ for $k>S_1(p_1)$,  since
the number of finite $T^{(j)}_i$'s is smaller or equal to
$S_1(p_1)-1.$ Similarly, given $P_2=p_2$, for $j\in
\mathcal{S}_2(p_2),$ $C_k^{(j)}=\emptyset$ for $k>S_2(p_2)$%, sincethe number of finite $T^{(j)}_i$'s is smaller or equal to $S_2(p_2)-1.$
. In addition, note that $C_k^{(j)}$ and $C_{k'}^{(j)}$ are disjoint
events for any $k\neq k'$ and
$\sum_{k=1}^{S_1(p_1)}\textmd{Pr}(C_k^{(j)}|P_1=p_1)=\sum_{k=1}^{S_2(p_2)}\textmd{Pr}(C_k^{(j)}|P_2=p_2)=1.$

The FDR for replicability analysis is
\begin{align}FDR&=\sum_{j=1}^m(1-H_{1j}H_{2j})\sum_{k=1}^m\frac{1}{k}\textmd{Pr}\left(j\in
\mathcal{S}_1\cap \mathcal{S}_2, r_j^{FDR}\leq \alpha,
C_k^{(j)}\right)\notag\\&=\sum_{j=1}^m(1-H_{1j}H_{2j})\sum_{k=1}^m\frac{1}{k}\textmd{Pr}\left(j\in
\mathcal{S}_1\cap \mathcal{S}_2, r_j\leq k\alpha,
C_k^{(j)}\right)\label{eqproc}\\&\leq\sum_{j=1}^m(1-H_{1j})\sum_{k=1}^m\frac{1}{k}\textmd{Pr}\left(j\in
\mathcal{S}_1\cap \mathcal{S}_2, r_j\leq k\alpha,
C_k^{(j)}\right)\label{ineqind}\\&+\sum_{j=1}^m(1-H_{2j})\sum_{k=1}^m\frac{1}{k}\textmd{Pr}\left(j\in
\mathcal{S}_1\cap \mathcal{S}_2, r_j\leq k\alpha,
C_k^{(j)}\right)\label{ineqind2}
\end{align}
where the equality in (\ref{eqproc}) follows from item 1, and the
inequality in (\ref{ineqind}) follows from the fact that
$1-H_{1j}H_{2j}\leq2-H_{1j}-H_{2j}$ for all $j\in\{1,\ldots,m\}.$ We
prove that for $(p_1, p_2)$ arbitrary fixed, the following
inequalities hold for conditional expectations.
\begin{align}&\sum_{j=1}^m(1-H_{1j})\sum_{k=1}^m\frac{1}{k}\textmd{Pr}\left(j\in \mathcal{S}_1\cap \mathcal{S}_2, r_j\leq k\alpha,
C_k^{(j)}\,|\,P_2=p_2\right)\leq \alpha_1,
%\\&\sum_{j\in \mathcal{S}_2(p_2)}^m(1-H_{1j})\sum_{k=1}^m\frac{1}{k}\textmd{Pr}\left(j\in \mathcal{S}_1, r_j\leq k\alpha, C_k^{(j)}\,|\,P_2=p_2\right)\leq
%c\alpha,
\label{cond1}\\&\sum_{j=1}^m(1-H_{2j})\sum_{k=1}^m\frac{1}{k}\textmd{Pr}\left(j\in
\mathcal{S}_1\cap \mathcal{S}_2, r_j\leq k\alpha,
C_k^{(j)}\,|\,P_1=p_1\right)\leq \alpha_2.\label{cond2}%\\&\sum_{j\in
%\mathcal{S}_1(p_1)}^m(1-H_{2j})\sum_{k=1}^m\frac{1}{k}\textmd{Pr}\left(j\in
%\mathcal{S}_2, r_j\leq k\alpha, C_k^{(j)}\,|\,P_1=p_1\right)\leq
%(1-c)\alpha.\label{cond2}
\end{align} Note that since these
inequalities hold for all $p_1$ and $p_2,$ they yield that the upper
bounds in (\ref{cond1}) and (\ref{cond2}) hold for expressions in
(\ref{ineqind}) and (\ref{ineqind2}) respectively, therefore FDR for
replicability analysis is upper bounded by
$\alpha_1+\alpha_2=\alpha.$ Thus it remains to prove inequalities
(\ref{cond1}) and (\ref{cond2}). We now prove inequality
(\ref{cond1}).
\begin{align}
&\sum_{j=1}^m(1-H_{1j})\sum_{k=1}^m\frac{1}{k}\textmd{Pr}\left(j\in
\mathcal{S}_1\cap \mathcal{S}_2, r_j\leq k\alpha,
C_k^{(j)}\,|\,P_2=p_2\right)=\notag\\&\sum_{j\in
\mathcal{S}_2(p_2)}(1-H_{1j})\sum_{k=1}^{S_2(p_2)}\frac{1}{k}\textmd{Pr}\left(j\in
\mathcal{S}_1, r_j\leq k\alpha,
C_k^{(j)}\,|\,P_2=p_2\right)\leq\label{lemcondw}\\&\sum_{j\in
\mathcal{S}_2(p_2)}(1-H_{1j})\sum_{k=1}^{S_2(p_2)}\frac{1}{k}\textmd{Pr}\left(j\in
\mathcal{S}_1, P_{1j}\leq \frac{kc\alpha}{W_1^{(j)}},
C_k^{(j)}\,|\,P_2=p_2\right)\leq\label{unif}\\&\alpha_1\sum_{j\in
\mathcal{S}_2(p_2)}(1-H_{1j})\sum_{k=1}^{S_2(p_2)}E\left(
\frac{1}{W_1^{(j)}}\textbf{I}\left[C_k^{(j)}\right]\,|\,P_2=p_2\right)=\notag\\&\alpha_1\sum_{j\in
\mathcal{S}_2(p_2)}(1-H_{1j})E\left(\frac{1}{W_1^{(j)}}\sum_{k=1}^{S_2(p_2)}
\textbf{I}\left[C_k^{(j)}\right]\,|\,P_2=p_2\right)=\label{ckj}\\&\alpha_1\sum_{j\in
\mathcal{S}_2(p_2)}(1-H_{1j})E\left(\frac{1}{W_1^{(j)}}\,|\,P_2=p_2\right)\leq\label{lemcond2}
\\&\alpha_1\sum_{j\in
\mathcal{S}_2(p_2)}(1-H_{1j})\left(\frac{1}{\sum_{j\in
\mathcal{S}_2(p_2)}(1-H_{1j})}\right)=\alpha_1.\notag
\end{align}
The inequality in (\ref{lemcondw}) follows from condition
(\textit{ii}) of item 2. The inequality in (\ref{unif}) follows from
the fact that the distribution of $P_{1j}$ is uniform or
stochastically larger than uniform and $P_{1j}$ with $H_{1j}=0$ is
independent of all other $p$-values. The equality in (\ref{ckj})
follows from the fact that given $P_2=p_2,$
$\cup_{k=1}^{S_2(p_2)}C_k^{(j)}$ is the whole sample space
represented as a union of disjoint events (as discussed above),
therefore $\sum_{k=1}^{S_2(p_2)}
\textbf{I}\left[C_k^{(j)}\right]=1.$ The inequality in
(\ref{lemcond2}) follows from condition (\textit{ii}) of item 2,
inequality (\ref{wj1}). Thus we proved inequality (\ref{cond1}).
Inequality
(\ref{cond2}) is proved similarly%using condition (b) of item 2, inequality (\ref{wj2})
.
%THIS WAS PREVIOUSLY ITEM B IN THEOREM 3.3, SO REFERENCES IN THE PROOF NEED TO BE CORRECTED.
%XXXXXXXXXXXXXXXXXX NEW PROOF FOR DEPENDENCY XXXXXXXXXXXXXXXXXX
\noindent
\\\textbf{Proof of item 2 in Theorem \ref{indep}.} The proof is
similar to the proof of item 3 of Theorem S3.2 in the Supplementary
Material of \cite{Bogomolov13}. We give it below for completeness.
For $j,k\in\{1,\ldots,m\},$ we define $\tilde{C}_k^{(j)}$ as the
event in which if $\tilde{r}_j^{FDR}\leq \alpha,$ then the total
number of arbitrary-dependence FDR $r$-values which are at most
$\alpha$ is $k.$ Similarly to the proof of item 2 of Lemma
\ref{lem_fdr} it can be shown that the event $\widetilde{C}_k^{(j)}$
is defined on the space of $(P_1^{(j)}, P_2^{(j)})$ as in
(\ref{ckj-lemma}), where $T$-values are replaced by $\tilde{T}$-
values which are defined as follows.
\begin{align}\tilde{T}_{i}^{(j)}=\left\{
\begin{array}{cl}
\max\left(\frac{(\sum_{k=1}^{S_2^{(j)}+1}1/k)(S_2^{(j)}+1)p_{1i}}{c},
\frac{(\sum_{k=1}^{(S_1^{(j)}+1)}1/k)(S_1^{(j)}+1)p_{2i}}{1-c}\right)  & \text{if}\,\, i\in \mathcal{S}_1^{(j)}\cap \mathcal{S}_2^{(j)},\\
 \infty& \text{otherwise. } \\
\end{array} \right. \notag\end{align}
Note that $\tilde{T}_{i}^{(j)}=\tilde{r}_i$ for $i\in
\mathcal{S}_1^{(j)}\cap \mathcal{S}_2^{(j)},$ where the expression
for $\tilde{r}_i$ is given in (\ref{rtilda}). Similarly to the proof
of item 2 of Lemma \ref{lem_fdr}, it can be shown that given
$P_1=p_1,$ $\tilde{C}_k^{(j)}=\emptyset$ for $j\in
\mathcal{S}_1(p_1)$ and $k>S_1(p_1),$ and $\cup_{k=1}^{S_1(p_1)}
\tilde{C}_k^{(j)}$ is the whole sample space. Given $P_2=p_2,$
$\tilde{C}_k^{(j)}=\emptyset$ for $j\in \mathcal{S}_2(p_2)$ and
$k>S_2(p_2),$ and $\cup_{k=1}^{S_2(p_2)}\tilde{C}_k^{(j)}$ is the
whole sample space. In addition, $\tilde{C}_k^{(j)}$ and
$\tilde{C}_{k'}^{(j)}$ are disjoint events for any $k\neq k'.$

We obtain the following inequality for the FDR for replicability
analysis using derivations  (\ref{eqproc})-(\ref{ineqind2}) where we
replace $r_j^{FDR},$ $r_j$ and $C_k^{(j)}$ by $\tilde{r}_j^{FDR},$
$\tilde{r}_j$ and $\tilde{C}_k^{(j)},$ respectively.
\begin{align}FDR&
\leq\sum_{j=1}^m(1-H_{1j})\sum_{k=1}^m\frac{1}{k}\textmd{Pr}\left(j\in
\mathcal{S}_1\cap \mathcal{S}_2, \tilde{r}_j\leq k\alpha,
\tilde{C}_k^{(j)}\right)\notag\\&+\sum_{j=1}^m(1-H_{2j})\sum_{k=1}^m\frac{1}{k}\textmd{Pr}\left(j\in
\mathcal{S}_1\cap \mathcal{S}_2, \tilde{r}_j\leq k\alpha,
\tilde{C}_k^{(j)}\right)\label{fdr-dep-1}
\end{align}

We now find an upper bound for the first term of the sum in
(\ref{fdr-dep-1}).
 Let $p_2$ be arbitrary fixed. We define
 $\tilde{\alpha}_1=\alpha_1/(\sum_{i=1}^{S_2(p_2)}1/i).$ We shall prove that
\begin{align}\sum_{j=1}^m(1-H_{1j})\sum_{k=1}^m\frac{1}{k}\textmd{Pr}\left(j\in
\mathcal{S}_1\cap \mathcal{S}_2, \tilde{r}_j\leq k\alpha,
\tilde{C}_k^{(j)}\,|\,P_2=p_2\right)\leq
\alpha_1.\label{ineqcond1}\end{align} %We first show that the
%expression in (\ref{ineqcond1}) is upper bounded by a sum of two
%terms.
Note that
\begin{align}
&\sum_{j=1}^m(1-H_{1j})\sum_{k=1}^m\frac{1}{k}\textmd{Pr}\left(j\in
\mathcal{S}_1\cap \mathcal{S}_2, \tilde{r}_j\leq k\alpha,
\tilde{C}_k^{(j)}\,|\,P_2=p_2\right)=\sum_{j\in
\mathcal{S}_2(p_2)}(1-H_{1j})\times\notag\\&\sum_{k=1}^{S_2(p_2)}\frac{1}{k}\,\textmd{Pr}\left(
S_1\left(\sum_{i=1}^{S_1}1/i\right)p_{2j}\leq k\alpha_2, j\in
\mathcal{S}_1,
P_{1j}\leq\frac{k\tilde{\alpha}_1}{S_2(p_2)},
\tilde{C}_{k}^{(j)}\,|\,P_2=p_2\right)\notag&\\&\leq \sum_{j\in
\mathcal{S}_2(p_2)}(1-H_{1j})\sum_{k=1}^{S_2(p_2)}\frac{1}{k}\,\textmd{Pr}\left(
P_{1j}\leq\frac{k\tilde{\alpha}_1}{S_2(p_2)},
\tilde{C}_{k}^{(j)}\,|\,P_2=p_2\right).\notag
\end{align}
For each $j$ with ${H}_{1j}=0,$ $k\in\{1,\ldots,S_2(p_2)\},$ and
$l\in\{1,\ldots,k\},$ let us define:
\begin{align*}p_{jkl}=\textmd{Pr}\left(P_{1j}\in\left(\frac{(l-1)\tilde{\alpha}_1}{S_2(p_2)},
\frac{l\tilde{\alpha}_1}{S_2(p_2)} \right],
\tilde{C}_k^{(j)}| P_2=p_2\right).\end{align*} Note that for $j$
with ${H}_{1j}=0,$ $\textmd{Pr}\left(P_{1j}\leq x\right)\leq x$ for
all $x\in[0,1]$, in particular $\textmd{Pr}\left(P_{1j}=0\right)=0.$
Therefore, for each $j$ with ${H}_{1j}=0$ and
$k\in\{1,\ldots,S_2(p_2)\},$
\begin{align}
\textmd{Pr}\left(P_{1j}
\leq\frac{k\tilde{\alpha}_1}{S_2(p_2)},
\tilde{C}_{k}^{(j)}\,|\,P_2=p_2\right)=\sum_{l=1}^{k}p_{jkl}.\notag
\end{align}
Using this equality we obtain:
\begin{align}
&\sum_{j\in
\mathcal{S}_2(p_2)}(1-H_{1j})\sum_{k=1}^{S_2(p_2)}\frac{1}{k}\textmd{Pr}\left(P_{1j}
\leq\frac{k\tilde{\alpha}_1}{S_2(p_2)},\,
\tilde{C}_{k}^{(j)}\,|\,P_2=p_2\right)=\notag\\&\sum_{j\in
\mathcal{S}_2(p_2)}(1-H_{1j})\sum_{k=1}^{S_2(p_2)}\frac{1}{k}\sum_{l=1}^kp_{jkl}=\sum_{j\in
\mathcal{S}_2(p_2)}(1-H_{1j})\sum_{l=1}^{S_2(p_2)}\sum_{k=l}^{S_2(p_2)}\frac{1}{k}p_{jkl}\leq\notag\\&
\sum_{j\in\mathcal{S}_2(p_2)}(1-H_{1j})\sum_{l=1}^{S_2(p_2)}\sum_{k=l}^{S_2(p_2)}\frac{1}{l}p_{jkl}
\leq\sum_{j\in
\mathcal{S}_2(p_2)}(1-H_{1j})\sum_{l=1}^{S_2(p_2)}\frac{1}{l}\sum_{k=1}^{S_2(p_2)}p_{jkl}.\label{last3}
\end{align}
Since $\cup_{k=1}^{S_2(p_2)} C_k^{(j)}$ is a union of disjoint
events, we obtain for each $j$ with ${H}_{1j}=0$ and
$l\in\{1,\ldots,S_2(p_2)\}$:
\begin{align}
&\sum_{k=1}^{S_2(p_2)} p_{jkl}=\textmd{Pr}\left(P_{1j}\in
\left(\frac{(l-1)\tilde{\alpha}_1}{S_2(p_2)},
\frac{l\tilde{\alpha}_1}{S_2(p_2)} \right],\,
\cup_{k=1}^{S_2(p_2)}\tilde{C}_k^{(j)}\,|\,P_2=p_2\right)\notag\\&\leq\textmd{Pr}\left(P_{1j}\in
\left(\frac{(l-1)\tilde{\alpha}_1}{S_2(p_2)},
\frac{l\tilde{\alpha}_1}{S_2(p_2)} \right]\,|\,P_2=p_2\right)\notag\\&=
\textmd{Pr}\left(P_{1j}\leq\frac{l\tilde{\alpha}_1}{S_2(p_2)}\,|P_2=p_2\right)-\textmd{Pr}\left(P_{1j}\leq\frac{(l-1)\tilde{\alpha}_1}{S_2(p_2)}\,|\,P_2=p_2\right).\notag
\end{align}
Therefore for each $j$ with $H_{1j}=0$ we obtain:
\begin{align}
&\sum_{l=1}^{S_2(p_2)}\frac{1}{l}\sum_{k=1}^{S_2(p_2)}p_{jkl}\leq
\sum_{l=1}^{S_2(p_2)}\frac{1}{l}\left[\textmd{Pr}\left(P_{1j}\leq
\frac{l\tilde{\alpha}_1}{S_2(p_2)}\,|\,P_2=p_2\right)-\textmd{Pr}\left(P_{1j}\leq
\frac{(l-1)\tilde{\alpha}_1}{S_2(p_2)}\,|\,P_2=p_2\right)\right]\notag\\&=
\sum_{l=1}^{S_2(p_2)}\frac{1}{l}\textmd{Pr}\left(P_{1j}\leq
\frac{l\tilde{\alpha}_1}{S_2(p_2)}\,|\,P_2=p_2\right)-\sum_{l=0}^{S_2(p_2)-1}\frac{1}{l+1}\textmd{Pr}\left(P_{1j}\leq
\frac{l\tilde{\alpha}_1}{S_2(p_2)}\,|\,P_2=p_2\right)\notag
\\&=\sum_{l=1}^{S_2(p_2)-1}\left(\frac{1}{l}-\frac{1}{l+1}
\right)\textmd{Pr}\left(P_{1j}\leq\frac{l\tilde{\alpha}_1}{S_2(p_2)}\,|\,P_2=p_2\right)+\frac{1}{S_2(p_2)}\textmd{Pr}\left(P_{1j}\leq
\tilde{\alpha}_1\,|\,P_2=p_2\right)\notag\\&\leq\sum_{l=1}^{S_2(p_2)-1}\frac{1}{l+1}\left(\frac{\tilde{\alpha}_1}{S_2(p_2)}\right)+\frac{\tilde{\alpha}_1}{S_2(p_2)}=\left(\frac{\tilde{\alpha}_1}{S_2(p_2)}\right)\sum_{l=1}^{S_2(p_2)}\frac{1}{l}=\frac{\alpha_1}{S_2(p_2)}.\label{danilast3}
\end{align}
The inequality in (\ref{danilast3}) follows from the null
independence-across-studies condition and the fact that for $j$ with
$H_{1j}=0,$ $\textmd{Pr}(P_{1j}\leq x)\leq x$ for all $x\in[0,1].$
Combining (\ref{last3}) with (\ref{danilast3}) we obtain the
inequality in (\ref{ineqcond1}):
\begin{align}&\sum_{j\in
\mathcal{S}_2(p_2)}(1-H_{1j})\sum_{k=1}^{S_2(p_2)}\frac{1}{k}\textmd{Pr}\left(P_{1j}
\leq\frac{k\tilde{\alpha_1}}{S_2(p_2)},\,
\tilde{C}_{k}^{(j)}\,|\,P_2=p_2\right)\leq\alpha_1\frac{\sum_{j\in
\mathcal{S}_2(p_2)}(1-H_{1j})}{S_2(p_2)}\leq\alpha_1.\notag%\left(\frac{\tilde{\alpha}_1}{S_2(p_2)}\right)\sum_{l=1}^{S_2(p_2)}\frac{1}{l}\notag\\&=\tilde{\alpha}_1\frac{\sum_{j\in
%\mathcal{S}_2(p_2)}(1-H_{1j})}{S_2(p_2)}\sum_{l=1}^{S_2(p_2)}\frac{1}{l}\leq\tilde{\alpha}_1\sum_{l=1}^{S_2(p_2)}\frac{1}{l}=\alpha_1.\notag
\end{align}

We have proved that the inequality in (\ref{ineqcond1}) holds for
$p_2$ arbitrary fixed, therefore the first term of the sum in
(\ref{fdr-dep-1}) is upper bounded by $\alpha_1.$ Similarly it can
be proven that the second term of the sum in (\ref{fdr-dep-1}) is
upper bounded by $\alpha_2.$ It follows from (\ref{fdr-dep-1}) that
these two inequalities yield $FDR\leq \alpha_1+\alpha_2=\alpha.$

%XXXXXXXXXXXXXXXXXXXXXXX END OF NEW PROOF FOR DEPENDENCY XXXXXXXXXXXXXXXXXX
\section{Theoretical properties for Section \ref{sec-estthresholds}}\label{app-estthresholdsTheoretical}
We use the following lemma to justify the empirical selection of
$(t_1,t_2)$ for Procedure \ref{proc-FWER} based on Bonferroni.
\begin{lemma}\label{lemma-estthreshold}
Assume that $G_i(t_i), i \in \{1,2\}$  are monotone increasing
functions. Let  $(t^*_1,t^*_2)$ be the solution to the following two
equations:
\begin{equation}
t_1 = \frac{\alpha_1}{G_2(t_2)};  \quad t_2 =
\frac{\alpha_2}{G_1(t_1)}. \label{eqlemma}
\end{equation}
 Then there does not exist a pair $(t_1,t_2)$ that dominates $(t_1^*, t_2^*)$ in the following sense: $$\binom{\min\left(t_1, \frac{\alpha_1}{ G_2(t_2)}\right)}{\min\left(\frac{\alpha_2}{ G_1(t_1)},t_2  \right)} > \binom{\min\left(t_1^*, \frac{\alpha_1}{ G_2(t^*_2)}\right)}{\min\left(\frac{\alpha_2}{ G_1(t_1^*)},t_2^*  \right)},$$ where the strict inequality means that both coordinates are at least as large with  $(t_1,t_2)$ as with $(t_1^*, t_2^*)$, but at least one coordinate is strictly larger.
\end{lemma}
See Appendix \ref{app-lemmaproof} for a proof. Clearly, $S_1(t_1)$
and $|\mathcal{S}_{1,\lambda}(t_1)|\prtwo$ are increasing functions
of $t_1$, and similarly $S_2(t_2)$ and
$|\mathcal{S}_{2,\lambda}(t_2)|\prone$ are increasing functions of
$t_2$. From Lemma \ref{lemma-estthreshold} it follows that  the
choice $(t^*_1,t^*_2)$ in equations (\ref{eq-nonlin}) or
(\ref{sel-fwer-adapt}) is not dominated by any other choice of
$(t_1,t_2)$ in Procedure \ref{procfdrsym} and Procedure
\ref{proc-bonferroniadapt}, respectively. Therefore, we suggest
these data-dependent $(t^*_1,t^*_2)$.

Our next theorems state that the FWER and FDR of the non-adaptive
procedures using the above data-dependent thresholds for selection
are controlled under independence.
\begin{theorem}\label{thm-fwer-data-dep}
If the $p$-values from true null hypotheses within each study are
exchangeable, and each independent of all other $p$-values, then
Procedure \ref{proc-FWER} based on Bonferroni with selection
thresholds $(t_1,t_2) =(t^*_1,t^*_2)$ which are a single solution to
equations (\ref{eq-nonlin}) controls the FWER for replicability
analysis at level $\alpha.$
\end{theorem}
The proof of Theorem \ref{thm-fwer-data-dep} is given in Appendix
\ref{app-thm-fwer-data-dep}.

\begin{theorem}\label{thm-FDR-data-dep}
If the $p$-values from true null hypotheses within each study are
exchangeable, and each independent of all other $p$-values,
Procedure \ref{procfdrsym} with selection thresholds $(t_1,t_2)
=(t^*_1,t^*_2)$ which are a single solution to equations
(\ref{sel-fdr}) controls the FDR for replicability analysis at level
$\alpha$.
\end{theorem}
The proof of Theorem \ref{thm-FDR-data-dep} is given in Appendix
\ref{app-FDR-data-dep}.

\subsection{Proof of Lemma \ref{lemma-estthreshold}}\label{app-lemmaproof}
The proof is by contradiction. Suppose that there exists a pair
$(t_1^o, t_2^o)$ that dominates $(t^*_1,t^*_2)$, in the sense that
$$\binom{\min\left(t_1^o, \frac{\alpha_1}{G_2(t_2^o)}\right)}{\min\left(\frac{\alpha_2}{ G_1(t_1^o)},t_2^o  \right)} > \binom{\min\left(t_1^*, \frac{\alpha_1}{G_2(t^*_2)}\right)}{\min\left(\frac{\alpha_2}{ G_1(t_1^*)},t_2^*  \right)}.$$
 Then either the first coordinate or the second coordinate satisfy a strict inequality. Without loss of generality, assume
that the first coordinate satisfies a strict inequality, i.e.
\begin{align}&t_1^*= \frac{\alpha_1}{G_2(t^*_2)}<\min\left(t_1^o,
\frac{\alpha_1}{
G_2(t_2^o)}\right),\label{t1}\\&t_2^*=\frac{\alpha_2}{G_1(t_1^*)}\leq
\min\left(t_2^o, \frac{\alpha_2}{G_1(t_1^o)}\right)\label{t2}
.\end{align} It follows from (\ref{t1}) that
$\frac{\alpha_1}{G_2(t^*_2)}<\frac{\alpha_1}{ G_2(t_2^o)}$,
therefore using the fact that $G_2(t_2)$ is a monotone increasing
function we obtain that $t_2^o<t_2^*$. It follows from (\ref{t2})
that $t_2^0\geq t_2^*.$ A contradiction is thus reached.
%
%
%Two inequalities follow:
%\begin{enumerate}
%\item $\frac{\alpha_1}{G_2(t^*_2)}<\frac{\alpha_1}{ G_2(t_2^o)}$, therefore  $t_2^o<t_2^*$.
%\item $t_1^*<t_1^o \rightarrow \frac{\alpha_2}{ G_1(t_1^o)}<\frac{\alpha_2}{G_1(t_1^*)}$.
%\end{enumerate}
%A contradiction is thus reached, since
%$\min\left(\frac{\alpha_1}{G_1(t_1^o)},t_2^o
%\right)<\min\left(\frac{\alpha_2}{G_1(t_1^*)},t_2^*  \right)$.

\subsection{Proof of Theorem \ref{thm-fwer-data-dep}}\label{app-thm-fwer-data-dep}
Procedure \ref{proc-FWER} based on Bonferroni makes replicability
claims for features with indices in the set $\{j: P_{1j}\leq u_1,
P_{2j}\leq u_2\},$ where $u_1=\min(t_1, \alpha_1/S_2(t_2)),$
$u_2=\min(t_2, \alpha_2/S_1(t_1)).$ Obviously the choice of
selection thresholds $t_1^*, t_2^*$ solving the equations
(\ref{eq-nonlin}) leads to the rejection thresholds $(u_1^*,
u_2^*)=(t_1^*, t_2^*),$ i.e. satisfying
\begin{equation}\label{eq-bonf1}
u_1^*=\alpha_1/S_2(u_2^*), u_2^*=\alpha_2/S_1(u_1^*).
\end{equation}
Thus the FWER of Procedure \ref{proc-FWER} using $(t_1^*, t_2^*)$ is
bounded above by
$$\sum_{j=1}^m(1-H_{1j})\textmd{Pr}(P_{1j}\leq u_1^*, P_{2j}\leq u_2^*)+\sum_{j=1}^m(1-H_{2j})\textmd{Pr}(P_{1j}\leq u_1^*, P_{2j}\leq u_2^*).$$
We shall only show that the upper bound of the first sum is at most
$\alpha_1$, since the proof that the upper bound of the second sum
is at most $\alpha_2$ follows similarly.

For each $j\in \{1,\ldots, m\}$,  define $(u_1^{*(j)}, u_2^{*(j)})$
to be the solution of the equations
$$S_2(u_2)u_1= \alpha_1, [S_{1}^{(j)}(u_1)+1]u_2= \alpha_2.$$ Note
that $(u_1^{*(j)}, u_2^{*(j)})$ are independent of $P_{1j}$ and that
if $P_{1j}\leq u_1^*$ then $(u_1^*, u_2^*) = (u_1^{*(j)},
u_2^{*(j)})$.

Consider now $j$ with $\vec H_j \in \{(0,0), (0,1)\}$:
\begin{eqnarray}
&&\textmd{Pr}(P_{1j}\leq u_1^*, P_{2j}\leq u_2^*\vert P_{1}^{(j)}, P_{2}) \nonumber \\ &&= \textmd{Pr}(P_{1j}\leq u_1^{*(j)}, P_{2j}\leq u_2^{*(j)}, (u_1^*, u_2^*) = (u_1^{*(j)}, u_2^{*(j)}) \vert P_{1}^{(j)}, P_{2}) \nonumber \\
&&\leq  \textmd{Pr}( P_{1j}\leq \frac{\alpha_1}{S_2(u_2^{*(j)})}, P_{2j}\leq u_2^{*(j)} \vert P_{1}^{(j)}, P_{2}) \nonumber\\%\label{eq-bonf3}\\
&&=\frac{\alpha_1}{S_2(u_2^{*(j)})}\textbf{I}[P_{2j}\leq u_2^{*(j)}]
\label{eq-bonf4}
\end{eqnarray}
%The inequality (\ref{eq-bonf3}) follows from  inequality
%(\ref{eq-bonf2}).
The equality  (\ref{eq-bonf4}) follows from the fact that $H_{1j}
=0$,  so $P_{1j}$ has a  uniform distribution (or is stochastically
larger than uniform). Let $j_0\in \{j:  \vec H_j \in \{(0,0),
(0,1)\}\}$ be an arbitrary fixed index . It thus follows that
\begin{eqnarray}
&&\sum_{j=1}^m(1-H_{1j})\textmd{Pr}(P_{1j}\leq u_1^*, P_{2j}\leq u_2^*) = \sum_{j=1}^m(1-H_{1j})E\{\textmd{Pr}(P_{1j}\leq u_1^*, P_{2j}\leq u_2^*\vert P_{1}^{(j)}, P_{2})\} \nonumber\\
&&\leq \sum_{j=1}^m(1-H_{1j})E\left\{\frac{\alpha_1}{S_2(u_2^{*(j)})}\textbf{I}[P_{2j}\leq u_2^{*(j)}]\right\}\nonumber \\
&& = \sum_{j=1}^m(1-H_{1j}) E\left\{\frac{\alpha_1}{S_2(u_2^{*(j_0)})}\textbf{I}[P_{2j}\leq u_2^{*(j_0)}]\right\}\label{eq-bonf5} \\
&& \leq \alpha_1
E\left\{\frac{\sum_{j=1}^m(1-H_{1j})\textbf{I}[P_{2j}\leq
u_2^{*(j_0)}]}{S_2(u_2^{*(j_0)})}\right\}\leq \alpha_1.
\end{eqnarray}
The equality (\ref{eq-bonf5}) follows from the fact that the
distribution of $\left(P_1^{(j_0)}, P_2\right)$ is the same as that
of $\left(P_1^{(j)}, P_2\right)$ for every
for every $j$ with $\vec{H}_j\in\{(0,0), (0,1)\}$, since the $p$-values are assumed to be
independent and exchangeable under the null.

\subsection{Proof of Theorem \ref{thm-FDR-data-dep}}\label{app-FDR-data-dep}
Procedure \ref{procfdrsym} makes the replicability claims for
features with indices in the set $\{j: P_{1j}\leq u_1, P_{2j}\leq
u_2\},$ where $u_1=\min(t_1, R\alpha_1/S_2(t_2)),$ $u_2=\min(t_2,
R\alpha_2/S_1(t_1),$ and
$$R=\max\left\{r:
\sum_{j\in\mathcal{S}_1(t_1)\cap\mathcal{S}_2(t_2)}\textbf{I}\left[(p_{1j},
p_{2j})\leq\left(\frac{r\alpha_1}{S_2(t_2)},
\frac{r\alpha_2}{S_1(t_1)}\right)\right] = r\right\}.$$ Note that
$R\leq |\mathcal{S}_1(t_1)\cap \mathcal{S}_2(t_2)|.$ In addition,
when the choice of selection thresholds is $(t_1, t_2)=(t_1^*,
t_2^*),$ which satisfy
\begin{align}t_1^*=\frac{|\mathcal{S}_1(t_1^*)\cap
\mathcal{S}_2(t_2^*)|\alpha_1}{S_2(t_2^*)},
t_2^*=\frac{|\mathcal{S}_1(t_1^*)\cap
\mathcal{S}_2(t_2^*)|\alpha_2}{S_1(t_1^*)},\label{eqnonlinfdr}\end{align}
it holds that
$$\sum_{j\in\mathcal{S}_1(t_1^*)\cap\mathcal{S}_2(t_2^*)}\textbf{I}\left[(p_{1j},
p_{2j})\leq\left(\frac{|\mathcal{S}_1(t_1^*)\cap
\mathcal{S}_2(t_2^*)|\alpha_1}{S_2(t_2^*)},
\frac{|\mathcal{S}_1(t_1^*)\cap\mathcal{S}_2(t_2^*)|\alpha_2}{S_1(t_1^*)}\right)\right]
= |\mathcal{S}_1(t_1^*)\cap\mathcal{S}_2(t_2^*)|.$$ Therefore, the
choice of selection thresholds $t_1^*, t_2^*$ solving the equations
in (\ref{eqnonlinfdr}) leads to
$R=|\mathcal{S}_1(t_1^*)\cap\mathcal{S}_2(t_2^*)|,$ and to the
rejection thresholds $(u_1^*, u_2^*)=(t_1^*, t_2^*).$ Thus the FDR
of Procedure \ref{procfdrsym} using $(t_1^*, t_2^*)$ is bounded
above by
\begin{align}FDR&=E\left(\frac{\sum_{j=1}^m(1-H_{1j})\textbf{I}(P_{1j}\leq
u_1^*, P_{2j}\leq u_2^*)}{\sum_{j=1}^m \textbf{I}(P_{1j}\leq u_1^*,
P_{2j}\leq
u_2^*)}\right)\notag\\&+E\left(\frac{\sum_{j=1}^m(1-H_{2j})\textbf{I}(P_{1j}\leq
u_1^*, P_{2j}\leq u_2^*)}{\sum_{j=1}^m \textbf{I}(P_{1j}\leq u_1^*,
P_{2j}\leq u_2^*)}\right),\label{fdr-data-dep}\end{align} where
$u_1^*$ and $u_2^*$ satisfy
$$u_1^*=\frac{|\mathcal{S}_1(u_1^*)\cap \mathcal{S}_2(u_2^*)|\alpha_1}{S_2(u_2^*)}, u_2^*=\frac{|\mathcal{S}_1(u_1^*)\cap \mathcal{S}_2(u_2^*)|\alpha_2}{S_1(u_1^*)}.$$
We shall only show that the first term of the sum in
(\ref{fdr-data-dep}) is upper bounded by $\alpha_1.$ The second term
of the sum in (\ref{fdr-data-dep}) is upper bounded by $\alpha_2,$
which yields that the FDR is upper bounded by $\alpha.$ The proof
that the the second term of the sum in (\ref{fdr-data-dep}) is at
most $\alpha_2$ follows similarly and is therefore omitted.

For $j\in\{1,\ldots,m\},$ define $(u_1^{*(j)}, u_2^{*(j)})$ to be
the solution of the equations
\begin{eqnarray}
&&u_1=\frac{(\textbf{I}[P_{2j}\leq u_2]+|\mathcal
S_{1}^{(j)}(u_1)\cap\mathcal
S_{2}^{(j)}(u_2)|)\alpha_1}{S_2(u_2)}, \nonumber \\
&& u_2=\frac{(\textbf{I}[P_{2j}\leq u_2]+|\mathcal
S_{1}^{(j)}(u_1)\cap\mathcal
S_{2}^{(j)}(u_2)|)q_2}{1+S_{1}^{(j)}(u_1)},
\end{eqnarray}
%where $\mathcal S_{i(-j)}(t)$ is defined in Appendix
%\ref{app-notation} for $i=1,2$.
 Note that if $P_{1j}\leq u_1^*$, then $(u_1^*, u_2^*) = (u_1^{*(j)},
u_2^{*(j)})$, and $|\mathcal S_1(u_1^*)\cap
\mathcal{S}_2(u_2^*)|=\textbf{I}[P_{2j}\leq
u_2]+|\mathcal{S}_{1}^{(j)}(u_1^{*(j)})\cap
\mathcal{S}_{2}^{(j)}(u_2^{*(j)})|.$ In addition, both $(u_1^{*(j)},
u_2^{*(j)})$ and $\textbf{I}[P_{2j}\leq
u_2]+|\mathcal{S}_{1}^{(j)}(u_1^{*(j)})\cap
\mathcal{S}_{2}^{(j)}(u_2^{*(j)})|$ are independent of $P_{1j}.$
Therefore,
\begin{align}
&E\left(\frac{\sum_{j=1}^m(1-H_{1j})\textbf{I}[P_{1j}\leq u_1^*,
P_{2j}\leq u_2^*]}{\sum_{j=1}^m \textbf{I}[P_{1j}\leq u_1^*,
P_{2j}\leq u_2^*]}\,|\,P_1^{(j)},
P_2\right)\notag\\&=\sum_{j=1}^m(1-H_{1j})E\left(\frac{\textbf{I}[P_{1j}\leq
u_1^*, P_{2j}\leq u_2^*]}{|\mathcal{S}_1(u_1^*)\cap
\mathcal{S}_2(u_2^*)|}\,|\,P_1^{(j)}, P_2\right)\notag\\
& =\sum_{j=1}^m(1-H_{1j})E\left(\frac{\textbf{I}[P_{1j}\leq u_1^*,
P_{2j}\leq u_2^*, u_1^*= u_1^{*(j)}, u_2^*=
u_2^{*(j)}]}{\textbf{I}[P_{2j}\leq u_2^{*(j)}]
+|\mathcal{S}_{1}^{(j)}(u_1^{*(j)})\cap
\mathcal{S}_{2}^{(j)}(u_2^{*(j)})|}\,|\,P_1^{(j)}, P_2\right)\notag
\\&\leq \sum_{j=1}^m(1-H_{1j})E\left(\frac{\textbf{I}[P_{1j}\leq u_1^{*(j)}, P_{2j}\leq
u_2^{*(j)}]}{\textbf{I}[P_{2j}\leq
u_2^{*(j)}]+|\mathcal{S}_{1}^{(j)}(u_1^{*(j)})\cap
\mathcal{S}_{2}^{(j)}(u_2^{*(j)})|}\,|\,P_1^{(j)}, P_2\right)\notag
\\&=\sum_{j=1}^m(1-H_{1j})\frac{\textmd{Pr}(P_{1j}\leq u_1^{*(j)})
%\frac{[S_2(u_2^{*(-j)})^{(-j)}+1]u_1^{*(-j)}}{1+|\mathcal
%S_1(u_1^{*(-j)})^{(-j)}\cap\mathcal S_2(u_2^{*(-j)})^{(-j)}}|,
\textbf{I}[P_{2j}\leq u_2^{*(j)}]}{\textbf{I}[P_{2j}\leq
u_2^{*(j)}]+|\mathcal{S}_{1}^{(j)}(u_1^{*(j)})\cap
\mathcal{S}_{2}^{(j)}(u_2^{*(j)})|}\label{fdr-data-dep-ind}\\&\leq
\sum_{j=1}^m(1-H_{1j})\frac{u_1^{*(j)}\textbf{I}[P_{2j}\leq
u_2^{*(j)}]}{\textbf{I}[P_{2j}\leq
u_2^{*(j)}]+|\mathcal{S}_{1}^{(j)}(u_1^{*(j)})\cap
\mathcal{S}_{2}^{(j)}(u_2^{*(j)})|}\label{fdr-data-dep-ind2}\\&=
\alpha_1\sum_{j=1}^m(1-H_{1j})\frac{\textbf{I}[P_{2j}\leq
u_2^{*(j)}]}{S_2(u_2^{*(j)})}.\label{fdr-data-dep-unif}
\end{align}
The equality in (\ref{fdr-data-dep-ind}) follows from the fact that
$P_{1j}$ and $(P_1^{(j)}, P_2)$ are independent for any $j$ with
$H_{1j}=0.$ The inequality in (\ref{fdr-data-dep-ind2}) follows from
the fact that $P_{1j}$ has a distribution at least stochastically as
large as the uniform distribution for $j$ with $H_{1j}=0,$ and the
equality in (\ref{fdr-data-dep-unif}) follows from the definition of
$(u_1^{*(j)}, u_2^{*(j)}).$

Let $j_0$ be an arbitrary fixed index in $\{j:H_{1j}=0\}$. It thus
follows that
\begin{eqnarray}
&&E\left(\frac{\sum_{j=1}^m(1-H_{1j})\textbf{I}[P_{1j}\leq u_1^*,
P_{2j}\leq u_2^*]}{\sum_{j=1}^m \textbf{I}[P_{1j}\leq u_1^*,
P_{2j}\leq
u_2^*]}\right)=\nonumber\\&&E\left[E\sum_{j=1}^m(1-H_{1j})\left(\frac{\textbf{I}[P_{1j}\leq
u_1^*, P_{2j}\leq u_2^*]}{\sum_{j=1}^m \textbf{I}[P_{1j}\leq u_1^*,
P_{2j}\leq u_2^*]}\,|\,P_1^{(j)},
P_2\right)\right]\nonumber\\
&&\leq \alpha_1
\sum_{j=1}^m(1-H_{1j})E\left(\frac{\textbf{I}[P_{2j}\leq
u_2^{*(j)}]}{S_2(u_2^{*(j)})}\right) \nonumber \\
&&  =\alpha_1
E\left\{\frac{\sum_{j=1}^m(1-H_{1j})\textbf{I}[P_{2j}\leq
u_2^{*(j_0)}]}{S_2(u_2^{*(j_0)})}\right\}\label{eq-fdr5} \leq
\alpha_1.
\end{eqnarray}
The equality in (\ref{eq-fdr5}) follows from the fact that the
distribution of $\left(P_1^{(j_0)}, P_2\right)$ is the same as that
of $\left(P_1^{(j)}, P_2\right)$ for every $j$ with $H_{1j}=0,$
since the $p$-values are assumed to be independent and exchangeable
under the null. %Thus we have proved that the first term of the sum
%in (\ref{fdr-data-dep}) is upper bounded by $\alpha_1.$

 \end{document}